\documentclass[11pt,reqno]{amsart}
\usepackage[a4paper,textwidth=154mm,hcentering,top=28mm,bottom=28mm]{geometry}
\usepackage{amsmath,amssymb,amsthm,booktabs,hyperref}
\hypersetup{colorlinks=true,urlcolor=blue,linkcolor=blue,citecolor=blue,filecolor=blue,pdfborder={0 0 0}}
\numberwithin{equation}{section}
\newtheorem{theorem}{Theorem}[section]
\newtheorem{lemma}[theorem]{Lemma}
\newtheorem{proposition}[theorem]{Proposition}
\newtheorem{corollary}[theorem]{Corollary}
\theoremstyle{definition}

\theoremstyle{remark}
\newtheorem{remark}[theorem]{Remark}
\newtheorem{problem}[theorem]{Problem}
\newcommand{\E}{\mathbb E}
\title[Near-Gaussian counterexamples to entropy concavity]{Near-Gaussian counterexamples to the Ball--Nayar--Tkocz entropy concavity conjecture}
\author{Congyi Luo}
\address{School of Data Science, Fudan University, Shanghai 200433, China}
\email{cyluo24@m.fudan.edu.cn}
\keywords{Differential entropy, strongly log-concave densities, Gaussian smoothing, Hermite polynomials, interpolation inequalities}
\date{}
\begin{document}
\begin{abstract}
The Ball--Nayar--Tkocz conjecture asserts that, for independent real random variables $X,Y$ with a common log-concave density, the function $t\mapsto h(\sqrt tX+\sqrt{1-t}Y)$ is concave on $[0,1]$.
Gaussian distributions have this property. We construct strongly log-concave counterexamples arbitrarily close to the Gaussian and study their persistence under Gaussian smoothing.
Let $X$ have mean zero and variance one, let $G$ be an independent standard Gaussian, and let $q\in[0,1]$ denote the signal variance proportion in the smoothed variable $\sqrt qX+\sqrt{1-q}G$.
\textbf{In the asymmetric case, for every fixed $0<q\le1$, there are counterexamples for which the entropy of the weighted sum of two independent copies of the smoothed variable has strictly positive second derivative near the endpoints.}
\textbf{In the symmetric case, for every $4/7<q\le1$, there are counterexamples for which an unequally weighted sum has strictly greater entropy than the equally weighted sum.}
Both families have smooth, strictly positive densities and are strongly log-concave before and after smoothing. The density ratios $f/\varphi$ converge uniformly to one, and each derivative of fixed positive order converges uniformly to zero, where $\varphi$ is the standard Gaussian density.
The symmetric counterexamples can also match any prescribed finite number of Gaussian moments exactly.
The asymmetric construction uses an endpoint expansion of entropy, whereas the symmetric construction exploits the different decay rates of high-order Hermite perturbations under different weights.
Thus, densities violating concavity approach the Gaussian density in the strong sense described above.
The question of uniform concavity in the symmetric class for $0<q\le4/7$ remains open.
\end{abstract}
\maketitle
\section{Introduction and main theorems}\label{sec:introduction}
Let $X,Y$ be independent real random variables with common density $f$, and write
\[
 h(X)=-\int_{\mathbb R}f\log f,\qquad
 F_f(t)=h(\sqrt tX+\sqrt{1-t}Y),\quad 0\le t\le1.
\]
Throughout, logarithms are natural, and every entropy considered is required to have an absolutely integrable integrand $f\log f$.
Ball, Nayar, and Tkocz~\cite[Conjecture 2]{BNT} ask whether $F_f$ is concave whenever $f$ is log-concave.
This normalization keeps the variance of the weighted sum fixed. Since $F_f(t)=F_f(1-t)$, concavity necessarily implies
\begin{equation}\label{eq:necessary-midpoint}
 F_f(1/2)\ge\tfrac12\{F_f(t)+F_f(1-t)\}=F_f(t).
\end{equation}
For a Gaussian density $f$, the function $F_f$ is constant. Two questions therefore arise: does concavity hold in some neighborhood of the Gaussian, and does Gaussian smoothing with a fixed proportion suffice to ensure concavity? We give counterexamples to both assertions, distinguishing symmetric and asymmetric inputs.

Write $\varphi(x)=(2\pi)^{-1/2}e^{-x^2/2}$. For a density $f$ with mean zero and variance one, let
\[
 X_q=\sqrt qX+\sqrt{1-q}G,\qquad 0\le q\le1,
\]
where $X\sim f$ and $G$ is an independent standard Gaussian, and denote the density of $X_q$ by $\mathcal S_qf$.
The parameter $q$ is the fraction of the total variance contributed by the signal: $q=1$ gives the original distribution, whereas $q=0$ gives the standard Gaussian.
We study $F_{\mathcal S_qf}$, the entropy of the weighted sum of two independent copies of the smoothed input.
We prove that asymmetric counterexamples arbitrarily close to the Gaussian exist for every $0<q\le1$, and symmetric ones exist for every $4/7<q\le1$.
Both families approach the Gaussian uniformly in relative density and in derivatives of every fixed order, and remain strongly log-concave before and after smoothing.
The symmetric counterexamples can also match any prescribed finite number of Gaussian moments exactly.
The symmetric case $0<q\le4/7$ remains unresolved; we do not establish $4/7$ as a critical value for entropy concavity.

\subsection{Main theorems}
For a smooth function $u$, write
\[
 \|u\|_{C_b^m}=\max_{0\le j\le m}\sup_{x\in\mathbb R}|u^{(j)}(x)|,
 \qquad m=0,1,2,\ldots.
\]
A strictly positive $C^2$ density $f$ is called strongly log-concave if $(\log f)''\le-c$ for some $c>0$.
All convergence of relative densities below means convergence in $C_b^m$ for every fixed $m$.

\begin{theorem}[Asymmetric counterexamples: $0<q\le1$]\label{thm:asymmetric-all-noise}
Fix $0<q\le1$. There exists a sequence $(f_n)$ of strictly positive, smooth, asymmetric probability densities, each with mean zero and variance one, such that:
\begin{enumerate}
\item[(i)] \textnormal{\bfseries Strict convexity near the endpoint.} For every sufficiently large $n$, there exists $\delta_n\in(0,1/2)$ such that
\begin{equation}\label{eq:main-asymmetric-curvature}
 F_{\mathcal S_qf_n}''(t)>0,
 \qquad 0<t<\delta_n.
\end{equation}
\item[(ii)] \textnormal{\bfseries Strongly log-concave approximation of the Gaussian.} Both $f_n$ and $\mathcal S_qf_n$ are strongly log-concave, and
\[
 \left\|\frac{f_n}{\varphi}-1\right\|_{C_b^m}\longrightarrow0,
 \qquad
 \left\|\frac{\mathcal S_qf_n}{\varphi}-1\right\|_{C_b^m}\longrightarrow0,
 \qquad m=0,1,2,\ldots.
\]
\end{enumerate}
\end{theorem}
The sequence may depend on the fixed $q$. Inequality~\eqref{eq:main-asymmetric-curvature} directly contradicts concavity. Thus no fixed positive signal fraction guarantees concavity throughout the general strongly log-concave class.
The explicit construction is given in Section~\ref{sec:asymmetric-noise}; the required endpoint expansion is proved in Appendix~\ref{sec:endpoint}.

\begin{theorem}[Symmetric counterexamples: $4/7<q\le1$]\label{thm:scaled-hermite-counterexamples}
Fix an integer $K\ge1$. There exists a sequence $(f_n)$ of strictly positive, smooth, symmetric probability densities of variance one with the following properties.
\begin{enumerate}
\item[(i)] \textnormal{\bfseries The equally weighted sum does not maximize entropy.} For every $4/7<q\le1$, there are $t_q\in(0,1/2)$, independent of $n$, and an integer $N_q$ such that
\begin{equation}\label{eq:main-gap}
 F_{\mathcal S_qf_n}(t_q)>F_{\mathcal S_qf_n}(1/2),
 \qquad n\ge N_q.
\end{equation}
\item[(ii)] \textnormal{\bfseries Smooth approximation of every fixed order.} For every fixed $4/7<q\le1$ and integer $m\ge0$,
\begin{equation}\label{eq:main-smooth}
 \begin{aligned}
 \left\|\frac{f_n}{\varphi}-1\right\|_{C_b^m}&\longrightarrow0,\\
 \left\|\frac{\mathcal S_qf_n}{\varphi}-1\right\|_{C_b^m}&\longrightarrow0.
 \end{aligned}
\end{equation}
Moreover, for sufficiently large $n$, the densities before and after smoothing satisfy
\[
 (\log f_n)''<-\frac12,
 \qquad (\log\mathcal S_qf_n)''<-\frac12.
\]
\item[(iii)] \textnormal{\bfseries Exact matching of finitely many Gaussian moments.} For every $4/7<q\le1$, every $n$, and $j=0,1,\ldots,2K$,
\begin{equation}\label{eq:main-moments}
 \begin{aligned}
 \int_{\mathbb R}x^j f_n(x)\,dx
 &=\int_{\mathbb R}x^j\varphi(x)\,dx,\\
 \int_{\mathbb R}x^j(\mathcal S_qf_n)(x)\,dx
 &=\int_{\mathbb R}x^j\varphi(x)\,dx.
 \end{aligned}
\end{equation}
\end{enumerate}
The same sequence $(f_n)$ works for all these $q$; its choice depends only on $K$.
\end{theorem}

Inequality~\eqref{eq:main-gap} reverses the necessary condition~\eqref{eq:necessary-midpoint} for concavity.
Since $F_{\mathcal S_qf_n}(t_q)=F_{\mathcal S_qf_n}(1-t_q)$, it gives a strict violation at a pair of weights symmetric about the midpoint.
Taking $q=1$ gives symmetric counterexamples to the Ball--Nayar--Tkocz conjecture while retaining the approximation and moment constraints in (ii) and (iii).

\begin{remark}[Comparison weights and uniformity in the parameter]
The comparison weights in Theorem~\ref{thm:scaled-hermite-counterexamples} can be chosen on either side of the midpoint.
For every fixed $4/7<q\le1$, there exists $a_q>0$ such that every fixed $t\in(0,1)$ satisfying
\[
 0<\left|t-\frac12\right|<a_q
\]
has $F_{\mathcal S_qf_n}(t)>F_{\mathcal S_qf_n}(1/2)$ for all sufficiently large $n$.
The lower bound on the degree may depend on $q$ and the chosen $t$.
Furthermore, for every fixed $q_0\in(4/7,1]$, a common comparison weight and a common lower bound on the degree can be chosen for $q\in[q_0,1]$.
\end{remark}
Section~\ref{sec:consequences} expresses these approximation properties as boundary statements at the Gaussian density.

\subsection{Related results and the proofs}
Gaussian scale mixtures provide a class for which the proposed concavity holds. Such variables have the form $SG$, where $S>0$ is independent of a standard Gaussian $G$.
Eskenazis, Nayar, and Tkocz~\cite[Theorem 8]{ENT} proved that the equally weighted sum of independent copies of a Gaussian scale mixture maximizes weighted entropy whenever the relevant entropies are finite;
Eskenazis and Gavalakis~\cite[Theorem 1]{EG} further proved entropy concavity for weighted sums of two independent Gaussian scale mixtures.
Theorem~\ref{thm:scaled-hermite-counterexamples} shows that the same conclusion fails in the symmetric strongly log-concave class, even arbitrarily close to the Gaussian and under any prescribed finite number of exact Gaussian moment constraints.

The normalization of the weights also affects the direction of entropy comparisons. Yu~\cite[Theorem 1]{Yu} proved that, for independent copies of a log-concave variable,
\[
 h(\lambda X+(1-\lambda)Y)\ge h((X+Y)/2),\qquad 0\le\lambda\le1.
\]
Here the sum of the coefficients is fixed. We instead fix their squared sum, preserving the variance of the weighted sum; the proposed midpoint inequality is~\eqref{eq:necessary-midpoint}.

Madiman, Nayar, and Tkocz~\cite[Theorem 1]{MNT} studied comparisons involving a fixed copy and a normalized sum of several copies; their counterexample is not the two-copy weight comparison in~\eqref{eq:main-gap}.
Appendix~\ref{sec:main-family} gives a more elementary construction with a single cosine, showing directly how equal and unequal weights produce different entropies under convolution.

The asymmetric theorem relies on an endpoint expansion. For a density $g$ with mean zero and variance one whose relative perturbation of the Gaussian is a Schwartz function, let
\[
 \mu_3(g)=\int x^3g(x)\,dx,\qquad
 J_3(g)=\int g(x)\bigl((\log g)'(x)\bigr)^3\,dx.
\]
Under the assumptions of Appendix~\ref{sec:endpoint}, we prove
\[
 F_g''(t)=-\frac{\mu_3(g)J_3(g)}{16\sqrt t}+O(1),\qquad t\downarrow0.
\]
Thus $\mu_3(g)J_3(g)<0$ suffices to give strictly positive curvature near the endpoint.
The construction takes the form $f_\varepsilon=\varphi(1+\varepsilon u+\varepsilon^2d_qv)$:
the fourth and fifth Hermite components of $u$ preserve the first three moments and produce the required contribution to the third moment of the logarithmic derivative;
the third-order component of $v$ adjusts the third central moment.
For fixed $q>0$, we choose $d_q$ so that the smoothed density satisfies $\mu_3(g)>0$ and $J_3(g)<0$, and then choose $\varepsilon$ sufficiently small to retain strong log-concavity.
Section~\ref{sec:asymmetric-noise} gives this choice and controls the remainder.

For symmetric densities, $\mu_3(g)=J_3(g)=0$, so this endpoint term vanishes.
We instead use Hermite functions of high degree at a fixed spatial scale and cancel finitely many low-order moments.
The relative perturbation of the convolution density splits into terms linear and quadratic in the amplitude.
Mehler's formula gives an exact generating function for the linear term; a four-variable Gaussian integral and comparison of positive coefficients give upper and lower bounds for the quadratic term.
When $q>4/7$, the quadratic term dominates the linear term at the midpoint, whereas at a fixed nearby nonmidpoint weight it has additional exponential decay.
The linear term and the moment correction do not alter this comparison. Converting the norm comparison of relative perturbations into an entropy comparison gives~\eqref{eq:main-gap}.

Section~\ref{sec:boundary} summarizes the two families and the boundary properties at the Gaussian; Section~\ref{sec:open-problems} asks about uniform concavity, local concavity near the Gaussian, and midpoint curvature.
Appendices~\ref{sec:endpoint}--\ref{sec:asymmetric} establish the endpoint expansion and moment estimates used in the asymmetric theorem.
Appendices~\ref{sec:main-family}--\ref{sec:periodic-resonance} extend the single-cosine example to general periodic waveforms and derive a scaling limit near the midpoint;
Appendix~\ref{app:midpoint-curvature} then studies the second derivative at the midpoint directly.
The Gaussian interpolation theorem in Appendix~\ref{sec:interpolation} explains the appearance of $4/7$ in the comparison of linear and quadratic terms; its parameter boundary is not asserted to be a sufficient condition for entropy concavity.

\section{Gaussian smoothing and Hermite functions}
\label{sec:preliminaries}
The proofs of both main theorems reduce the sign of entropy differences or derivatives to explicit perturbation estimates.
We establish two-sided quadratic bounds for relative entropy and compute a generating function for Gaussian smoothing of Hermite functions.
The former reduces entropy comparisons at equal variance to comparisons of Gaussian $L^2$ norms; the latter gives an exact expression for these norms as the degree varies.
Bounds for Hermite derivatives ensure that the small perturbations below remain positive and strongly log-concave.
We first introduce the entropy bounds and smoothing operator common to both constructions, and then the Hermite estimates needed for the symmetric construction. Section~\ref{sec:asymmetric-noise} uses only the first part and the finite-degree polynomials in the appendices.
For convenience, $\|\cdot\|_2$ below denotes the norm in $L^2(\mathbb R,\varphi(x)\,dx)$, and $G$ always denotes a standard Gaussian variable.
If $p=\varphi(1+r)$ is a probability density with $\|r\|_\infty\le\eta<1$, define
$D(p\Vert\varphi)=\int p\log(p/\varphi)$. For comparisons of relative entropy and quadratic divergence under bounded density ratios, see~\cite[Section IV]{SV}; the constants needed here follow from the integral form of Taylor's formula below.
\begin{lemma}\label{lem:entropy-quadratic}
Under the preceding assumptions,
\[
 \frac{\|r\|_{L^2(\varphi\,dx)}^2}{2(1+\eta)}\le D(p\Vert\varphi)
 \le\frac{\|r\|_{L^2(\varphi\,dx)}^2}{2(1-\eta)}.
\]
If $p$ has mean zero and variance one, then $h(p)=\frac12\log(2\pi e)-D(p\Vert\varphi)$.
\end{lemma}
\begin{proof}
Since $\int\varphi r=0$, Taylor's formula with integral remainder gives
\[
 D(p\Vert\varphi)=\int\varphi(x)r(x)^2
   \left\{\int_0^1\frac{1-s}{1+sr(x)}\,ds\right\}\,dx.
\]
The inner integral lies between $1/[2(1+\eta)]$ and $1/[2(1-\eta)]$.
The positive upper and lower bounds on $p/\varphi$ also bound $p|\log p|$ by a Gaussian density times a quadratic polynomial, so the entropy is finite.
Substituting $\log\varphi=-\log(2\pi)/2-x^2/2$ into the definition of relative entropy gives the last identity.
\end{proof}

We use the probabilists' convention for Hermite polynomials:
\[
 e^{xz-z^2/2}=\sum_{n=0}^\infty H_n(x)\frac{z^n}{n!}.
\]
Let $\varphi_s$ denote the centered Gaussian density of variance $s$. For $0\le\rho\le1$, define
\[
 T_\rho g(x)=\mathbb E g(\rho x+\sqrt{1-\rho^2}G).
\]
If $X$ has density $\varphi(1+g)$, then $\sqrt qX+\sqrt{1-q}G$ has density
$\varphi(1+T_{\sqrt q}g)$. Indeed, let $U,V$ be standard Gaussians with correlation $\rho$. The conditional law
$U\mid V=x$ is $N(\rho x,1-\rho^2)$. Weighting the joint Gaussian density by $1+g(U)$ gives $V$ the relative density $1+T_\rho g$.
Exchanging $U,V$ also shows that $T_\rho$ is self-adjoint on $L^2(\varphi)$.
Conditional Jensen's inequality and invariance of the Gaussian marginals give
$\int|T_\rho g|^2\varphi\le\int T_\rho(|g|^2)\varphi=\int|g|^2\varphi$.
The definition also gives $T_\rho T_\sigma=T_{\rho\sigma}$, because the two independent Gaussian noises combine to have variance $1-\rho^2\sigma^2$.
The operator contracts both $L^2(\varphi)$ and $L^\infty$ norms, and satisfies
$\|(T_\rho g)^{(m)}\|_\infty\le\rho^m\|g^{(m)}\|_\infty$.
When the corresponding derivative of $g$ is bounded, dominated convergence justifies differentiation under the expectation.
The Hermite generating function further gives $T_\rho H_j=\rho^jH_j$: applying $T_\rho$ to $e^{zx-z^2/2}$ gives $e^{\rho zx-\rho^2z^2/2}$, and one compares coefficients. This diagonalization of Gaussian convolution also appears in Abbe--Zheng~\cite{AZ}.

\begin{lemma}\label{lem:hermite-bounds}
Let $a>0$ and
\[
 h_{n,a}(x)=\sqrt2e^{-ax^2/2}\frac{H_n(\sqrt{2a}x)}{\sqrt{n!}}.
\]
Then $\|h_{n,a}\|_\infty\le B_n:=\sqrt{2e}(n+1)^{1/4}$, and the norm of every fixed-order derivative
grows at most polynomially in $n$. In particular,
$\|h_{n,a}''\|_\infty\le a(2n+2)B_{n+2}$.
\end{lemma}
\begin{proof}
Mehler's identity (\cite[Eq. 18.18.28]{DLMF}, converted using $H_n^{\rm prob}(x)=2^{-n/2}H_n^{\rm phys}(x/\sqrt2)$) reads
\[
 \sum_{j\ge0}\frac{r^jH_j(x)H_j(y)}{j!}
 =\frac1{\sqrt{1-r^2}}
 \exp\!\left(\frac{2rxy-r^2(x^2+y^2)}{2(1-r^2)}\right),\quad |r|<1.
\]
Setting $x=y=\sqrt{2a}\,u$ gives
\[
 \sum_{j\ge0}r^j h_{j,a}(u)^2
 =\frac2{\sqrt{1-r^2}}
       \exp\!\left(-a\frac{1-r}{1+r}u^2\right).
\]
For $n\ge1$, take $r=n/(n+1)$. The inequalities $r^{-n}\le e$ and
$1-r^2\ge1/(n+1)$ yield the bound. For $n=0$, the function $h_{0,a}(x)=\sqrt2e^{-ax^2/2}$ satisfies the same bound.
The Hermite recurrence gives
\[
 h_{n,a}'=\sqrt{a/2}\bigl(\sqrt n\,h_{n-1,a}
                         -\sqrt{n+1}\,h_{n+1,a}\bigr).
\]
Terms with negative indices are understood to be zero. Iterating $m$ times gives at most $2^m$ terms, each with coefficient at most $(a/2)^{m/2}(n+m)^{m/2}$. Thus, for $n+m\ge1$,
\[
 \|h_{n,a}^{(m)}\|_\infty\le(2a)^{m/2}(n+m)^{m/2}B_{n+m}.
\]
In particular,
\[
 h_{n,a}''=\frac a2\{\sqrt{n(n-1)}h_{n-2,a}-(2n+1)h_{n,a}
 +\sqrt{(n+1)(n+2)}h_{n+2,a}\}.
\]
Using $\sqrt{n(n-1)}\le n$, $\sqrt{(n+1)(n+2)}\le n+2$, and monotonicity of $B_j$ gives the stated second-derivative bound.
\end{proof}

\begin{lemma}\label{lem:linear-mehler}
For $0\le r\le1$, in a neighborhood of $z=0$,
\[
 \sum_{n\ge0}z^n\|T_{\sqrt r}h_{n,a}\|_2^2
 =\frac2{\sqrt{(1+a)^2-a^2r^2}}
   \frac1{\sqrt{(1-\tau_a(r)z)(1-\nu_a(r)z)}},
\]
where
\[
 \tau_a(r)=\frac{1-a+ar}{1+a-ar},\qquad
 \nu_a(r)=\frac{a-1+ar}{1+a+ar}.
\]
\end{lemma}
\begin{proof}
Let $U,V$ be standard Gaussians with correlation $r$. Then
$\|T_{\sqrt r}h_{n,a}\|_2^2=\mathbb E h_{n,a}(U)h_{n,a}(V)$.
Integrating Mehler's formula and using the independent coordinates
$S=(U+V)/\sqrt2$ and $D=(U-V)/\sqrt2$, of variances $1+r$ and $1-r$, respectively, gives
\[
 \frac2{\sqrt{1-z^2}}\,
 \mathbb E\exp\!\left[-\frac a2\frac{1-z}{1+z}S^2
                     -\frac a2\frac{1+z}{1-z}D^2\right].
\]
Evaluating the one-dimensional Gaussian integrals, the square of the denominator is
\[
 \bigl(1+a(1+r)+[1-a(1+r)]z\bigr)
 \bigl(1+a(1-r)+[a(1-r)-1]z\bigr),
\]
Factoring gives the formula. To justify interchanging summation and integration, take $|z|\le z_0<1$.
The preceding lemma and the Cauchy--Schwarz inequality give
$\sum_{n\ge0}|z|^n|h_{n,a}(U)h_{n,a}(V)|\le\sum_{n\ge0}z_0^nB_n^2<\infty$,
Thus dominated convergence justifies the interchange. Both sides are analytic near the origin and therefore have the same coefficients.
\end{proof}

\begin{lemma}\label{lem:gaussian-coefficient}
Let $a>0$, and let the centered Gaussian vector $W\in\mathbb R^d$ have covariance matrix $C\ge0$. Then
\[
 \mathbb E\exp\!\left(-\frac a2W^\mathsf TW
                +\sqrt{2a}\,z^\mathsf TW-\frac12z^\mathsf Tz\right)
 =\frac{\exp\!\left(\frac12z^\mathsf T
 [2aC(I+aC)^{-1}-I]z\right)}{\sqrt{\det(I+aC)}}.
\]
\end{lemma}
\begin{proof}
When $C>0$, complete the square in the Gaussian integral. The tilted covariance is
$(C^{-1}+aI)^{-1}=C(I+aC)^{-1}$, which gives the identity.
For a general positive semidefinite $C$, apply the result to $C+\delta I$ and let $\delta\downarrow0$.
This limit can be taken under the expectation: restrict the complex parameter $z$ to any fixed compact set and use
$-a|W|^2/2+\sqrt{2a}\operatorname{Re}(z)^\mathsf TW\le|\operatorname{Re}z|^2$.
The modulus of the integrand has a bound independent of $W,\delta$. Couple $W_\delta=W+\sqrt\delta G_d$, where $G_d$ is an independent standard Gaussian vector, and apply dominated convergence.
\end{proof}

\section{Counterexamples for asymmetric inputs}\label{sec:asymmetric-noise}
In the asymmetric case, we seek strictly positive entropy curvature near the endpoint after a fixed amount of Gaussian smoothing.
For every fixed $q\in(0,1]$, we show that arbitrarily small smooth strongly log-concave perturbations achieve this.
The construction adjusts the third central moment and the third moment of the logarithmic derivative separately, making the former positive and the latter negative; the expansion in Appendix~\ref{sec:endpoint} then gives strict convexity near the endpoint.
Unlike the symmetric construction, whose oscillation degree tends to infinity, this construction only adjusts the amplitudes of two fixed functions.
Fix $q\in(0,1]$ and take
\begin{align}
 u(x)&=e^{-x^2/2}\{H_4(\sqrt2x)+H_5(\sqrt2x)\},\label{eq:u}\\
 v(x)&=e^{-x^2/2}H_3(\sqrt2x).\label{eq:v}
\end{align}

Here $H_j$ is the probabilists' Hermite polynomial. Write $u_q=T_{\sqrt q}u$ and $v_q=T_{\sqrt q}v$.

\begin{proof}[Proof of Theorem~\ref{thm:asymmetric-all-noise}]
Put $w=1-q/2$ and $z=x/\sqrt w$. Convolution identities for Gaussian derivatives give
\begin{align*}
 u_q(x)&=e^{-(1-w)z^2/2}
 \left[\frac{q^2}{4\sqrt2\,w^{5/2}}H_4(z)
       +\frac{q^{5/2}}{8w^3}H_5(z)\right],\\
 v_q(x)&=e^{-(1-w)z^2/2}\frac{q^{3/2}}{4w^2}H_3(z).
\end{align*}
More explicitly, setting $W_j(x)=e^{-x^2/2}H_j(\sqrt2x)$ yields
\[
 \varphi W_j=(-1)^j2^{-(j+1)/2}\varphi_{1/2}^{(j)}.
\]
After scaling by $\sqrt q$ and convolving with $\varphi_{1-q}$, the derivative term becomes
\[
 (-1)^j2^{-(j+1)/2}q^{j/2}\varphi_w^{(j)}.
\]

Using $\varphi_w^{(j)}=(-1)^jw^{-j/2}\varphi_wH_j(x/\sqrt w)$, we obtain
\[
 T_{\sqrt q}W_j(x)=\frac{q^{j/2}}{2^{(j+1)/2}w^{(j+1)/2}}
 e^{-(1-w)z^2/2}H_j(z),
\]
For $q=1$, this is understood as the identity without added noise. Taking $j=3,4,5$ gives the displayed formulas.
All these functions and their derivatives are bounded, and belong to the Schwartz class when $q>0$.
Let $K_q=\int x(u_q')^2\varphi$. Differentiating the even and odd parts and retaining only the cross term gives
\[
 K_q=\frac{60\sqrt2\,q^{9/2}P(q)}{(2-q)^{11/2}(q+2)^{13/2}},
\]
\[
 P(q)=3q^6-39q^5+264q^4-752q^3+1664q^2-1408q+1024.
\]
To evaluate this integral, write
\[
 A_4=\frac{q^2}{4\sqrt2w^{5/2}},\quad
 A_5=\frac{q^{5/2}}{8w^3},\quad D=\frac{q+2}{2},\quad
 B_j=H_j'-\frac q2 zH_j.
\]
After the substitution $x=\sqrt w\,z$, the cross term between the even and odd parts gives
\[
 K_q=\frac{2A_4A_5}{\sqrt{2\pi}}\int zB_4(z)B_5(z)e^{-Dz^2/2}\,dz.
\]
Here
\begin{align*}
 B_4(z)&=-\tfrac q2z^5+(4+3q)z^3-(12+3q/2)z,\\
 B_5(z)&=-\tfrac q2z^6+(5+5q)z^4-(30+15q/2)z^2+15.
\end{align*}
Writing $zB_4B_5=\sum_{j=1}^6b_{2j}z^{2j}$ and applying the Gaussian even-moment formula term by term gives
\[
 \sum_{j=1}^6b_{2j}(2j-1)!!D^{6-j}=\frac{15}{32}P(q),\qquad
 K_q=2A_4A_5D^{-13/2}\frac{15}{32}P(q),
\]
which is the asserted expression in rational powers.
To determine its sign, expand $P$ in the Bernstein basis as
$P(q)=\sum_{j=0}^6 p_j\binom6j q^j(1-q)^{6-j}$, where
\[
 (p_0,\ldots,p_6)=(1024,2368/3,3328/5,3076/5,9272/15,1331/2,756).
\]
All coefficients are strictly positive, so $K_q>0$.

Take $d_q=K_q/(2q^{3/2})$ and define
\[
 f_\epsilon=\varphi(1+\epsilon u+\epsilon^2d_qv),\qquad
 \mathcal S_qf_\epsilon=\varphi(1+\epsilon u_q+\epsilon^2d_qv_q).
\]
Lemma~\ref{lem:moments} gives $\int x^ju\varphi=\int x^jv\varphi=0$ for $0\le j\le2$, so both densities have exactly unit mass, mean zero, and variance one.
Indeed, self-adjointness of $T_{\sqrt q}$ and the binomial formula give, for $W=u,v$,
\[
 \int x^jT_{\sqrt q}W(x)\varphi(x)\,dx
 =\sum_{\ell=0}^j\binom j\ell q^{\ell/2}(1-q)^{(j-\ell)/2}
 \mathbb EG^{j-\ell}\int x^\ell W(x)\varphi(x)\,dx.
\]
Thus $\int x^ju_q\varphi=0$ for $0\le j\le3$, and
$\int x^3v_q\varphi=3q^{3/2}/2$. The third moment of the smoothed density is
\[
 \mu_3(\mathcal S_qf_\epsilon)=\tfrac32d_qq^{3/2}\epsilon^2=\tfrac34K_q\epsilon^2>0.
\]
Identity~\eqref{eq:score-exact} remains valid for this perturbation. Take
\[
 M_q=\max\{1,\|u\|_{C_b^2},d_q\|v\|_{C_b^2},
                 \|u_q\|_{C_b^2},d_q\|v_q\|_{C_b^2}\}.
\]
Gaussian smoothing preserves the vanishing moments. Applying Lemma~\ref{lem:general-score-remainder} with $U=u_q$ and $V=d_qv_q$ gives
\[
 |J_3(\mathcal S_qf_\epsilon)+\tfrac32K_q\epsilon^2|\le62M_q^3\epsilon^3.
\]
Hence $J_3(\mathcal S_qf_\epsilon)<0$ whenever $0<\epsilon<\min\{(8M_q)^{-1},K_q/(100M_q^3)\}$.
The same amplitude restriction makes the uniform norms of both relative perturbations and their second derivatives less than $1/4$.
If $r$ denotes either relative perturbation, then
\[
 (\log(\varphi(1+r)))''
 \le-1+\frac{\|r''\|_\infty}{1-\|r\|_\infty}<-2/3.
\]
Proposition~\ref{prop:endpoint} now gives strictly positive curvature near the endpoint.
Let $\epsilon$ run through any sequence tending to zero. Uniform convergence of each fixed-order derivative to zero, before and after smoothing, follows from the representation as a linear combination of fixed Schwartz functions.
\end{proof}

\section{Perturbations preserving finitely many Gaussian moments}\label{sec:construction-main}
For a symmetric density, both the third central moment and the third moment of the logarithmic derivative vanish, so the preceding endpoint criterion no longer gives a counterexample.
We now construct symmetric densities whose midpoint entropy is strictly smaller.
The first step is to satisfy positivity, strong log-concavity, and finite moment constraints simultaneously: for any fixed $K$, a finite-dimensional correction of the following high-degree Hermite perturbations matches the Gaussian moments through order $2K$ exactly and approaches the Gaussian in every fixed-order relative norm.
The moment correction decays exponentially with the degree; its effect will be controlled in the proof of the strict entropy gap.

Fix constants independent of $q$,
\[
 a=\frac7{3\sqrt5},\qquad b=\frac{a-1}{a+1}=\frac{47-21\sqrt5}{2}>0,
\]
and let $n$ tend to infinity through even integers, with
\[
 h_n(x)=\sqrt2 e^{-ax^2/2}\frac{H_n(\sqrt{2a}x)}{\sqrt{n!}},
 \qquad \epsilon_n=e^{-\sqrt n}.
\]
By Lemma~\ref{lem:hermite-bounds}, $\|h_n\|_\infty\le B_n=\sqrt{2e}(n+1)^{1/4}$, and the norm of each fixed-order derivative grows at most polynomially in $n$. In particular,
\[
 \|h_n''\|_\infty\le a(2n+2)B_{n+2}.
\]
We choose the moment correction in a finite-dimensional space of low-degree Hermite functions.

Let $\psi_j(x)=\sqrt2e^{-x^2/2}H_{2j}(\sqrt2x)/\sqrt{(2j)!}$, $0\le j\le K$, and choose coefficients $c_{n,j}$ such that
\[
 \widetilde h_n=h_n-\sum_{j=0}^K c_{n,j}\psi_j,
 \qquad \int x^{2\ell}\widetilde h_n(x)\varphi(x)dx=0
 \quad(0\le\ell\le K).
\]
Set $A_{\ell j}=\int x^{2\ell}\psi_j(x)\varphi(x)\,dx$.
Rodrigues' formula $H_j\varphi=(-1)^j\varphi^{(j)}$ and $2j$ integrations by parts give
\[
 A_{\ell j}=\frac{2^{-\ell}}{\sqrt{(2j)!}}
 \int y^{2\ell}H_{2j}(y)\varphi(y)\,dy.
\]
This vanishes if $\ell<j$ and equals $2^{-j}\sqrt{(2j)!}$ if $\ell=j$.
Every boundary term is the limit at infinity of a polynomial times a Gaussian function, and hence vanishes.
Thus $A$ is an invertible lower triangular matrix, and the equation $Ac_n=(\int x^{2\ell}h_n\varphi)_{\ell=0}^K$ uniquely determines the correction coefficients.

These coefficients decay exponentially in $n$. Integrating the Hermite generating function against the Gaussian gives
\[
 \int \varphi(x)e^{-ax^2/2}e^{\sqrt{2a}xz-z^2/2}dx
 =\frac1{\sqrt{1+a}}e^{bz^2/2}.
\]
After inserting $x^{2\ell}$, completing the square writes the right-hand side as
\[
 \frac{e^{bz^2/2}}{\sqrt{1+a}}\,
 \mathbb E\left(\frac{G}{\sqrt{1+a}}+\frac{\sqrt{2a}}{1+a}z\right)^{2\ell}.
\]
The expectation is a fixed polynomial containing only even powers, of degree at most $2\ell$.
For its $z^{2j}$ term, coefficient extraction introduces, relative to the $j=0$ term, the additional factor
$(2/b)^j(n/2)!/(n/2-j)!\le C_jn^j$ for $n\ge2K$. Extracting the coefficient of $z^n$ and using
$\sqrt{n!}/[2^{n/2}(n/2)!]=O(n^{-1/4})$, we obtain
\[
 \left|\int x^{2\ell}h_n\varphi\right|
 \le C_K n^K b^{n/2}\quad(\ell\le K).
\]
Since the inverse matrix is fixed and finite-dimensional, for every fixed $m$ this yields
\begin{equation}\label{eq:scaled-correction-small}
 \|\widetilde h_n-h_n\|_{C_b^m}\le C_{K,m}n^K b^{n/2}.
\end{equation}
Define $f_n=\varphi(1+\epsilon_n\widetilde h_n)$.
For each fixed $m$, the preceding estimates give $\|\widetilde h_n\|_{C_b^m}\le C_{K,m}n^{c_{K,m}}$, with exponent $c_{K,m}$ independent of $n$.
Since $-\sqrt n+c_{K,m}\log n\to-\infty$, we have $\epsilon_n\|\widetilde h_n\|_{C_b^m}\to0$.
Thus $1+\epsilon_n\widetilde h_n>1/2$ eventually, and uniform convergence of the relative density and all its fixed-order derivatives follows.
The zeroth-moment correction gives unit mass, the second-moment correction gives variance one, and evenness makes every odd moment vanish.
Logarithmic differentiation gives
\[
 (\log f_n)''\le-1+
 \frac{\epsilon_n\|\widetilde h_n''\|_\infty}
 {1-\epsilon_n\|\widetilde h_n\|_\infty}<-1/2
\]
for sufficiently large $n$. Now set $Z=\sqrt qX_n+\sqrt{1-q}G$. For $j\le2K$, independence gives
\[
 \mathbb EZ^j=\sum_{\ell=0}^j\binom j\ell q^{\ell/2}(1-q)^{(j-\ell)/2}
     \mathbb EX_n^\ell\,\mathbb EG^{j-\ell}.
\]
Each moment of $X_n$ can be replaced by that of an independent standard Gaussian, so the right-hand side is exactly the $j$th standard Gaussian moment.
At $q=0,1$, this identity is interpreted by continuity of the binomial expression.
Set $u_{q,n}=\epsilon_nT_{\sqrt q}\widetilde h_n$, so the smoothed density is $\varphi(1+u_{q,n})$.
Derivative contraction of the operator gives $\|u_{q,n}\|_{C_b^m}\le\epsilon_n\|\widetilde h_n\|_{C_b^m}\to0$, and
\[
 (\log\mathcal S_qf_n)''\le-1+
 \frac{\epsilon_n\|\widetilde h_n''\|_\infty}{1-\epsilon_n\|\widetilde h_n\|_\infty}<-1/2.
\]
These conclusions hold uniformly for $q\in[0,1]$. Moreover, the identity
$(\log(\varphi(1+u)))''+1=u''/(1+u)-(u')^2/(1+u)^2$ also gives
\[
 \sup_{0\le q\le1}\|(\log\mathcal S_qf_n)''+1\|_\infty\longrightarrow0.
\]
Evenness is preserved because $n$ is even and the Gaussian noise is symmetric.
Henceforth the degree is an even integer sufficiently large for the preceding positivity and curvature bounds to hold. Discarding finitely many initial terms and reindexing gives the sequence in the main theorem; subsequent estimates retain the original even degree as their index.

\section{Weighted convolution and the linear term}\label{sec:linear}
Having constructed densities satisfying the moment conditions, we determine how far their weighted convolutions depart from the Gaussian.
We decompose the output relative density exactly into terms linear and quadratic in the perturbation amplitude, and prove an exponential upper bound for the norm of the linear term.
The bound is given in~\eqref{eq:scaled-L-bound}; comparison with the midpoint lower bound for the quadratic term in the next section determines when the linear term is negligible.
Since the two smoothed copies are independent, their combined Gaussian noise has variance $1-q$. The full weighted convolution can therefore be represented by a conditional expectation of three independent standard Gaussians.

All norms in this section are in $L^2(\varphi)$. Temporarily use the uncorrected $h_n$ and write
\[
 L_t=T_{\sqrt{qt}}h_n+T_{\sqrt{q(1-t)}}h_n,
\]
Let $M_t(x)$ be the conditional expectation of $h_n(G_1)h_n(G_2)$ under three independent standard Gaussians, given
$\sqrt{qt}G_1+\sqrt{q(1-t)}G_2+\sqrt{1-q}G_3=x$. More explicitly, condition on the weighted sum being $x$ and expand
$(1+\epsilon_n\widetilde h_n(G_1))(1+\epsilon_n\widetilde h_n(G_2))$
The two linear conditional expectations are
$T_{\sqrt{qt}}\widetilde h_n(x)$ and $T_{\sqrt{q(1-t)}}\widetilde h_n(x)$.
A tilde denotes the corresponding conditional expectation with $h_n$ replaced by $\widetilde h_n$. The density of the weighted sum is then
\[
 p_{q,n,t}=\varphi(1+\epsilon_n\widetilde L_t+\epsilon_n^2\widetilde M_t).
\]
Define
\[
 \tau(r)=\frac{1-a+ar}{1+a-ar},\qquad
 \nu(r)=\frac{a-1+ar}{1+a+ar}.
\]
Lemma~\ref{lem:linear-mehler} gives
\begin{equation}\label{eq:scaled-linear-generating}
 \sum_{j\ge0}z^j\|T_{\sqrt r}h_j\|_2^2
 =\frac2{\sqrt{(1+a)^2-a^2r^2}}
 \frac1{\sqrt{(1-\tau(r)z)(1-\nu(r)z)}}.
\end{equation}
For $2/7\le r\le1$, we have $0<\nu(r)\le\tau(r)\le1$; equality $\nu(r)=\tau(r)$ holds precisely when $r=2/7$, and $\tau(r)=1$ precisely when $r=1$.
Indeed,
\[
 \tau(r)-\nu(r)=\frac{2(1-a^2+a^2r^2)}{(1+a-ar)(1+a+ar)},
 \qquad 1-\tau(r)=\frac{2a(1-r)}{1+a-ar}.
\]
Together with $a^2=49/45$ and $a>1$, these identities give the stated signs and equality cases.
Put $A_j=4^{-j}\binom{2j}{j}$. Then $(1-z)^{-1/2}=\sum A_jz^j$, and
$\sum_{j=0}^n A_jA_{n-j}=1$, since the product of the two series is $(1-z)^{-1}$.
Thus the $n$th coefficient of the generating function is at most its constant prefactor times $\tau(r)^n$, namely
\[
 \|T_{\sqrt r}h_n\|_2^2
 \le\frac2{\sqrt{(1+a)^2-a^2r^2}}\tau(r)^n.
\]
If $0\le r_1\le r_2$ and $r_2>0$, then
$T_{\sqrt{r_1}}=T_{\sqrt{r_1/r_2}}T_{\sqrt{r_2}}$, so the norm at the smaller parameter does not exceed the norm at the larger one.
Combining this with the triangle inequality, if $w=\max(t,1-t)$ and $q\ge4/7$, then
\begin{equation}\label{eq:scaled-L-bound}
 \|L_t\|_2^2\le
 \frac8{\sqrt{(1+a)^2-a^2q^2w^2}}\tau(qw)^n.
\end{equation}

\section{Positive coefficient comparison for the quadratic term}\label{sec:quadratic}
To produce the required entropy gap, the quadratic term must have different sizes at the midpoint and at the comparison weight.
Fix $q\in[4/7,1]$. We prove that, for every fixed $t\ne1/2$ sufficiently close to the midpoint, the ratio of the quadratic norm to its midpoint value tends to zero.
The midpoint lower bound and the nonmidpoint upper bound are given in~\eqref{eq:scaled-M-lower} and~\eqref{eq:scaled-M-other}, respectively.
We express the squared norm as a four-variable Gaussian integral and make a complex change of variables so that the generating function has nonnegative coefficients.
At the midpoint we retain one term with equal degree in each variable; away from it we bound the sum of coefficients. This gives different exponential decay rates.

Put $v=(\sqrt t,\sqrt{1-t})^\mathsf T$ and $S=\sqrt q\,v^\mathsf T(G_1,G_2)+\sqrt{1-q}G_3$. Take two conditionally independent copies with the conditional law of $(G_1,G_2)$ given $S$. Each copy has unconditional covariance $I_2$, and their cross covariance is $qvv^\mathsf T$, since the conditional mean is $\sqrt q\,vS$. The four Gaussian variables therefore have covariance
\[
 C_t=\begin{pmatrix}I_2&qvv^\mathsf T\\qvv^\mathsf T&I_2\end{pmatrix},
\]
Its eigenvalues are $1+q,1-q,1,1$, including the degenerate Gaussian case $q=1$. Lemma~\ref{lem:gaussian-coefficient} and the Hermite generating function give the exact identity
\begin{equation}\label{eq:scaled-four-coefficient}
 \|M_t\|_2^2=
 \frac{4(n!)^2}{\Delta}
 [z_1^nz_2^nz_3^nz_4^n]\exp\left(\frac12z^\mathsf T K_tz\right),
 \qquad \Delta=(1+a)\sqrt{(1+a)^2-a^2q^2},
\end{equation}
Here $K_t=2aC_t(I_4+aC_t)^{-1}-I_4$, and $[z_1^n\cdots z_4^n]$ denotes the coefficient of the indicated monomial.
The prefactor follows from normalization: the product of four $h_n$ contributes $4/(n!)^2$, while extracting coefficients from four Hermite generating functions contributes $(n!)^4$, giving $4(n!)^2$.
Also, $\sqrt{\det(I_4+aC_t)}=(1+a)\sqrt{(1+a)^2-a^2q^2}=\Delta$.

The following parameters are determined by the eigenvalues of this matrix; $\kappa$ will be the exponential decay factor of the midpoint quadratic term. Write
\[
 \lambda_\pm=\frac{a(1\pm q)-1}{a(1\pm q)+1},\quad
 \alpha=\frac{\lambda_++\lambda_-}{2}-b<0,\quad
 \beta=\frac{\lambda_+-\lambda_-}{2}>0,\quad \kappa=-\lambda_->0.
\]
Make the substitutions $z_1,z_2\mapsto iz_1,iz_2$ and $z_3,z_4\mapsto-iz_3,-iz_4$. The extracted coefficient is unchanged, since the total phase is $i^{2n}(-i)^{2n}=1$. The resulting real symmetric matrix is
\[
 K'_t=\begin{pmatrix}-bI_2-\alpha vv^\mathsf T&\beta vv^\mathsf T\\
 \beta vv^\mathsf T&-bI_2-\alpha vv^\mathsf T\end{pmatrix}.
\]
At $t=1/2$, the diagonal, within-group off-diagonal, and cross-group entries are, respectively,
\[
 d=-b-\alpha/2,\qquad e=-\alpha/2,\qquad c=\beta/2,
 \qquad d+e+2c=\kappa.
\]
These numbers are strictly positive for every $q\in[4/7,1]$, because
\[
 -\alpha=\frac{2a^2q^2}{(1+a)((1+a)^2-a^2q^2)}
\]
increases with $q$, while at $q=4/7$,
$d=(203\sqrt5-453)/38>0$. Thus all entries of $K'_t$ remain strictly positive in a neighborhood of the midpoint, allowing coefficientwise upper and lower bounds.

We first prove the midpoint lower bound. Put $Q(z)=z^\mathsf T K'_{1/2}z/2$, so $Q(1,1,1,1)=2\kappa$. After division by $(2\kappa)^{2n}$, the coefficients of $Q^{2n}$ describe a multinomial distribution from $2n$ independent choices among ten undirected edges, with loops allowed. Each of the four vertices has expected degree $n$. Choose the count $k_e$ of each within-group edge as a nearest even integer to $ne/\kappa$, and the count $k_c$ of each cross-group edge as a nearest integer to $nc/\kappa$, and set the count of each loop to
\[
 k_d=(n-k_e-2k_c)/2.
\]
The ten category probabilities are $d/(4\kappa)$ for each of four loops, $e/(2\kappa)$ for each of two within-group edges, and $c/(2\kappa)$ for each of four cross-group edges.
They sum to one, and the corresponding expected counts are $nd/(2\kappa)$, $ne/\kappa$, and $nc/\kappa$.
Since $n,k_e$ are even, $k_d$ is an integer, and
\[
 |k_e-ne/\kappa|\le1,\quad |k_c-nc/\kappa|\le1/2,\quad
 |k_d-nd/(2\kappa)|\le1.
\]
These counts are positive for sufficiently large $n$. Their sum is $4k_d+2k_e+4k_c=2n$, and each vertex has degree $2k_d+k_e+2k_c=n$.

To estimate this multinomial probability, put $N=2n$, denote the ten probabilities by $p_i>0$, and write the chosen counts as $N_i=Np_i+\delta_i$, where $|\delta_i|\le1$ and $\sum\delta_i=0$.
Stirling's bounds~\cite{Robbins} give absolute positive constants $C_1,C_2$ such that
$C_1\sqrt{k}(k/e)^k\le k!\le C_2\sqrt{k}(k/e)^k$ for $k\ge1$. Hence
\[
 \frac{N!\prod p_i^{N_i}}{\prod N_i!}
 \ge C N^{1/2}\prod N_i^{-1/2}
 \exp\left\{-\sum_i N_i\log\frac{N_i}{Np_i}\right\}.
\]
Take $n$ large enough that $|\delta_i/(Np_i)|\le1/2$.
Using $(1+u)\log(1+u)=u+O(u^2)$ and $\sum_i\delta_i=0$, the absolute value of the sum in the exponent is at most
$C\sum_i\delta_i^2/(Np_i)=O_q(N^{-1})$.
Since $N_i/N$ is bounded above and below by positive constants, this probability is at least $c_qN^{1/2-10/2}=c_qN^{-9/2}$.
On the stated compact parameter interval, the ten probabilities are uniformly bounded away from zero, so the constants and degree threshold can also be chosen uniformly. Combining
\[
 \frac{(n!)^2 4^n}{(2n)!}\asymp\sqrt n
\]
with~\eqref{eq:scaled-four-coefficient} gives
\begin{equation}\label{eq:scaled-M-lower}
 \|M_{1/2}\|_2^2\ge c_q n^{-4}\kappa^{2n}.
\end{equation}
The selected term in the exponential generating function has coefficient exactly
\[
 \frac{(d/2)^{4k_d}e^{2k_e}c^{4k_c}}
 {(k_d!)^4(k_e!)^2(k_c!)^4}.
\]
Appendix~\ref{sec:finite-certificate} uses this positive term directly.

For a fixed nearby weight $t\ne1/2$, all coefficients remain nonnegative. The required coefficient is therefore at most the sum of coefficients of the same total degree. A direct calculation gives
\[
 Q_t(1,1,1,1)=2\kappa_t,\qquad
 \kappa_t=\kappa-(\kappa+b)\left(\frac12-\sqrt{t(1-t)}\right)<\kappa.
\]
Consequently,
\begin{equation}\label{eq:scaled-M-other}
 \|M_t\|_2^2\le
 \frac4\Delta\frac{(n!)^2}{(2n)!}(2\kappa_t)^{2n}
 \le C_q\sqrt n\,\kappa_t^{2n}.
\end{equation}
It follows that $\|M_t\|_2/\|M_{1/2}\|_2\to0$.

\section{A strict entropy gap}\label{sec:comparison}
The preceding sections estimate the linear and quadratic terms. We must still exclude cancellation and show that moment corrections do not change the sign of the entropy gap.
For $q>4/7$, we prove that the entropy at a fixed nonmidpoint weight is strictly greater than at the midpoint, completing Theorem~\ref{thm:scaled-hermite-counterexamples}.
The estimates also permit a common comparison weight and degree threshold on any compact parameter interval $[q_0,1]\subset(4/7,1]$.
The midpoint quadratic norm is bounded below by a constant times $n^{-2}\kappa^n$, while the linear norm is bounded above by a constant times $\tau(q/2)^{n/2}$.

Thus a sufficient condition for the linear term to be negligible relative to the midpoint quadratic term is
\[
 \kappa^2>\tau(q/2).
\]
Let $u=a(1-q)$, so $\kappa=(1-u)/(1+u)$. Putting the terms over a common denominator gives
\[
 \kappa^2-\tau(q/2)=\frac{2a\{-1+3q/2+a^2(1-q/2)(1-q)^2\}}
 {[1+a(1-q)]^2[1+a-aq/2]}.
\]
The denominator is positive, and the polynomial in the numerator satisfies
\begin{align*}
 -1+\frac32q+a^2(1-q/2)(1-q)^2
 &=-\frac{(7q-4)(7q^2-24q+2)}{90}.
\end{align*}
On $q\in[4/7,1]$, the second quadratic factor is strictly negative. Hence the strict inequality holds precisely for $q>4/7$. At $q=4/7$, $\kappa=(3-\sqrt5)/2$ and $\tau(q/2)=\kappa^2$.

\begin{proof}[Proof of Theorem~\ref{thm:scaled-hermite-counterexamples}]
Fix $q>4/7$. By continuity, choose a fixed $t\ne1/2$ near the midpoint such that every entry of $K'_t$ is positive and $\tau(q\max(t,1-t))<\kappa^2$. Equations~\eqref{eq:scaled-L-bound}--\eqref{eq:scaled-M-other} and $\epsilon_n=e^{-\sqrt n}$ give
\[
 \frac{\|L_t\|_2+\|L_{1/2}\|_2}
 {\epsilon_n\|M_{1/2}\|_2}\to0,
 \qquad \frac{\|M_t\|_2}{\|M_{1/2}\|_2}\to0.
\]
More explicitly, let $\gamma=\sqrt{\tau(q\max(t,1-t))}/\kappa<1$ and $\theta=\kappa_t/\kappa<1$.
The two ratios above are at most $C_qn^2e^{\sqrt n}\gamma^n$ and $C_qn^{9/4}\theta^n$, respectively.
For any fixed $\xi\in(0,1)$ and real $c$,
$\log(n^ce^{\sqrt n}\xi^n)=c\log n+\sqrt n+n\log\xi\to-\infty$, so both ratios tend to zero.

Moment corrections do not change these comparisons. Contraction of conditional expectation and expansion of the difference of products give
\[
 \|\widetilde L_t-L_t\|_2\le2\|\widetilde h_n-h_n\|_\infty,
\]
\[
 \|\widetilde M_t-M_t\|_2
 \le2B_n\|\widetilde h_n-h_n\|_\infty
       +\|\widetilde h_n-h_n\|_\infty^2.
\]
Here $\sqrt b<\kappa$ throughout $q\in[4/7,1]$: $\kappa$ is increasing, $b<1/40$, and $\kappa(4/7)>3/8$. Thus all these errors from~\eqref{eq:scaled-correction-small}, divided by $\epsilon_n\|M_{1/2}\|_2$, still tend to zero.

With $q,n$ fixed, abbreviate $p_t=p_{q,n,t}$ and put $r_t=p_t/\varphi-1$. The conditional expectation representation also gives
\[
 \sup_t\|r_t\|_\infty\le
 \eta_n=2\epsilon_n\|\widetilde h_n\|_\infty+
 \epsilon_n^2\|\widetilde h_n\|_\infty^2\to0.
\]
At the midpoint, the reverse triangle inequality yields
\[
 \left|\|r_{1/2}\|_2-\epsilon_n^2\|M_{1/2}\|_2\right|
 \le\epsilon_n\|\widetilde L_{1/2}\|_2
       +\epsilon_n^2\|\widetilde M_{1/2}-M_{1/2}\|_2.
\]
The right-hand side divided by $\epsilon_n^2\|M_{1/2}\|_2$ tends to zero. Applying the triangle inequality at the comparison weight gives
\[
 \|r_{1/2}\|_2\sim\epsilon_n^2\|M_{1/2}\|_2,
 \qquad \|r_t\|_2=o(\epsilon_n^2\|M_{1/2}\|_2).
\]
The integral Taylor bounds in Lemma~\ref{lem:entropy-quadratic},
\[
 \frac{\|r\|_2^2}{2(1+\eta_n)}
 \le D(\varphi(1+r)\Vert\varphi)
 \le\frac{\|r\|_2^2}{2(1-\eta_n)},
\]
imply $D(p_t\Vert\varphi)/D(p_{1/2}\Vert\varphi)\to0$. The lower bound on the midpoint norm makes the denominator strictly positive, so eventually
\[
 D(p_t\Vert\varphi)<D(p_{1/2}\Vert\varphi).
\]
Every weighted sum has variance exactly one. Therefore
\[
 F_{\mathcal S_qf_n}(t)-F_{\mathcal S_qf_n}(1/2)
 =D(p_{1/2}\Vert\varphi)-D(p_t\Vert\varphi)>0.
\]
Exchanging the two independent copies gives $F(t)=F(1-t)$, contradicting Jensen's inequality for concavity between $t$ and $1-t$.

Finally, fix $q_0\in(4/7,1]$. The continuous function $\kappa(q)^2-\tau(q/2)$ has a positive minimum on $[q_0,1]$.
Uniform continuity permits a common $t\ne1/2$ such that all entries of $K'_t$ are positive throughout the interval and
$\sup_{q\in[q_0,1]}\sqrt{\tau(q\max(t,1-t))}/\kappa(q)<1$.
Moreover, $\kappa_t(q)/\kappa(q)<1$ and $\sqrt b/\kappa(q)<1$ have uniform upper bounds strictly below one on the same compact interval.
The polynomial factors, moment-correction constants, and constants in the positive coefficient lower bound are uniformly controlled.
All ratios therefore converge to zero uniformly, giving a common degree threshold.
For fixed $q$, the same continuity argument permits every fixed $t$ in a punctured neighborhood of the midpoint; the degree threshold may still depend on that value.
\end{proof}

\section{Consequences of the main theorems}\label{sec:consequences}
The weighted-sum entropy of the Gaussian is constant, whereas the counterexamples in Theorem~\ref{thm:scaled-hermite-counterexamples} approach it in relative norms of every order. To express this boundary property, fix $K\ge1$ and let $\mathcal A_K$ consist of densities satisfying:
$f$ is strictly positive, even, and smooth, $(\log f)''\le-1/2$, $f/\varphi-1\in C_b^\infty$,
$\inf_x f(x)/\varphi(x)>0$, and the moments of $f$ through order $2K$ agree with the standard Gaussian moments.
Here $C_b^\infty=\bigcap_{m\ge0}C_b^m$. Give $\mathcal A_K$ the relative topology determined by the seminorms
$\|f/\varphi-g/\varphi\|_{C_b^m}$.

For fixed $q$, let
\[
 \mathcal C_{q,K}=\{f\in\mathcal A_K:F_{\mathcal S_qf}\text{ is concave on }[0,1]\},
\]
\[
 \mathcal M_{q,K}=\{f\in\mathcal A_K:F_{\mathcal S_qf}(t)\le F_{\mathcal S_qf}(1/2)
                         \text{ for all }t\in[0,1]\}.
\]
All these entropies are finite: positive upper and lower bounds on the relative density are preserved by smoothing and conditional expectation, so Gaussian tails control the entropy integrals.

\begin{corollary}\label{cor:gaussian-boundary}
For every $K\ge1$ and every $4/7<q\le1$, the standard Gaussian density $\varphi$ belongs to the relative boundaries of both
$\mathcal C_{q,K}$ and $\mathcal M_{q,K}$.
\end{corollary}
\begin{proof}
Gaussian smoothing and weighted addition with unit squared sum preserve the standard Gaussian distribution, so $F_{\mathcal S_q\varphi}$ is constant.
Thus $\varphi\in\mathcal C_{q,K}\subset\mathcal M_{q,K}$.
The densities supplied by the main theorem eventually belong to $\mathcal A_K\setminus\mathcal M_{q,K}$ and converge to $\varphi$ in every fixed-order relative norm.
Every relative neighborhood of $\varphi$ therefore meets the complements of both sets. Since $\varphi$ itself belongs to both, it is a boundary point of each.
\end{proof}

Thus, in the specified relative topology, the Gaussian density is not an interior point of the set where concavity holds.
The limiting Gaussian entropy curve is constant and satisfies the concavity inequality with equality; the strict entropy gaps along the approximating sequence tend to zero with the perturbation amplitude.
Indeed, uniform convergence of the relative densities and Lemma~\ref{lem:entropy-quadratic} imply uniform convergence to zero of the relative entropies of the weighted sums, and hence of the counterexample gaps.
The matching order $2K$ is fixed before the construction, and the approximating sequence may depend on $K$.

\begin{corollary}\label{cor:necessary-noise}
Suppose a fixed $q\in[0,1]$ has the property that, for every symmetric, smooth, strongly log-concave density $f$ of variance one,
$F_{\mathcal S_qf}$ is concave. Then $q\le4/7$.
Consequently, if smoothing is written as $X+\sigma G$, a necessary condition for concavity uniformly over all such signals is
$\sigma^2\ge\frac34\operatorname{Var}(X)$.
\end{corollary}
\begin{proof}
If $q>4/7$, the main theorem directly supplies a density contradicting the assumption.
For a signal of variance $v>0$, center $X+\sigma G$ and divide by $\sqrt{v+\sigma^2}$. The signal variance fraction is then
$q=v/(v+\sigma^2)$. This common scaling adds only a weight-independent constant to the entropy curve and therefore does not affect concavity.
Solving $v/(v+\sigma^2)\le4/7$ gives the stated necessary condition.
\end{proof}

\section{Conclusions}\label{sec:boundary}
We have studied the Ball--Nayar--Tkocz entropy concavity conjecture: if $X,Y$ are independent real random variables with a common log-concave density, is the function
\[
 t\longmapsto h\!\left(\sqrt t\,X+\sqrt{1-t}\,Y\right)
\]
concave on $[0,1]$? The square-root weights keep the variance of the weighted sum fixed. For Gaussian inputs this entropy curve is constant and hence concave. We have further asked whether arbitrary proximity to the Gaussian, or Gaussian smoothing with a fixed proportion, suffices to guarantee the same property.

For an input of mean zero and variance one, let $q$ be the signal variance fraction in $\sqrt q\,X+\sqrt{1-q}\,G$, where $G$ is an independent standard Gaussian. We proved:
\begin{itemize}
\item \textbf{Asymmetric inputs: for every $0<q\le1$ there are counterexamples arbitrarily close to the Gaussian, with strictly convex entropy curves near the endpoint.}
\item \textbf{Symmetric inputs: for every $4/7<q\le1$ there are counterexamples arbitrarily close to the Gaussian for which an unequally weighted sum has strictly greater entropy than the equally weighted sum.}
\end{itemize}
Both families remain strongly log-concave before and after smoothing, with uniform convergence of relative densities and their derivatives of each fixed order. The symmetric family can also match any prescribed finite number of Gaussian moments exactly.
Thus, in these parameter ranges, every relative neighborhood of the Gaussian contains densities violating concavity, even under strong log-concavity and arbitrarily accurate uniform approximation of relative densities and all their fixed-order derivatives.

The Gaussian is consequently not an interior point of the corresponding concavity set; Section~\ref{sec:consequences} establishes this boundary property and the necessary noise lower bound in the symmetric class. Equality at the Gaussian limit is consistent with strict violations along the sequence, since the strict entropy gaps also tend to zero. For symmetric inputs, the range $0<q\le4/7$ remains unresolved: $4/7$ is the boundary of the auxiliary interpolation estimate, not an established critical value for entropy concavity. The remaining question is whether symmetry allows some positive signal fraction to restore concavity, and whether such restoration is possible at least near the Gaussian.

\section{Further questions}\label{sec:open-problems}
The main theorems leave a central question: can symmetry allow Gaussian smoothing with a positive signal fraction to restore concavity?
We have given asymmetric counterexamples for every positive signal fraction, while the symmetric range $0<q\le4/7$ remains unresolved.
We distinguish uniform concavity, local properties near the Gaussian, and midpoint curvature.

\subsection{Is there a positive range of uniform concavity for symmetric inputs?}
\begin{problem}\label{prob:uniform-noise}
Determine all $q\in(0,4/7]$ for which $F_{\mathcal S_qf}$ is concave for every symmetric log-concave density $f$ of variance one.
In particular, does any $q>0$ have this property?
\end{problem}
If such a positive signal fraction exists, symmetry provides a uniform concavity guarantee unavailable in the general log-concave class.
If none exists, symmetric counterexamples persist regardless of the amount of Gaussian noise, provided a positive signal fraction remains.
At $q=0$ the output is always Gaussian, so the question is whether this purely Gaussian case extends to any positive signal fraction.

The value $4/7$ arises here from the exact parameter boundary of an auxiliary interpolation inequality.
For $0<q\le4/7$, this inequality prevents the midpoint quadratic term from dominating the linear term for perturbations satisfying the corresponding small-norm condition.
It neither determines the sign of the second entropy derivative nor excludes other constructions.
Turning this limitation of the method into a concavity theorem, or constructing counterexamples beyond it, are therefore two distinct approaches to the question.

\subsection{Can counterexamples still be arbitrarily close to the Gaussian?}
Failure of uniform concavity does not imply that counterexamples occur near the Gaussian.
Our symmetric counterexamples admit relative approximation of every order and arbitrary finite moment matching. Whether these properties can be retained simultaneously in the remaining parameter range is a stronger question.
\begin{problem}\label{prob:local-gaussian}
Fix $q\in(0,4/7]$ and an integer $K\ge1$. Do there exist $f_n\in\mathcal A_K$ such that
\[
 \|f_n/\varphi-1\|_{C_b^m}\longrightarrow0
 \qquad\text{for every fixed }m\ge0,
\]
and $F_{\mathcal S_qf_n}$ fails to be concave for every $n$?
\end{problem}
Here $\mathcal A_K$ is the symmetric strongly log-concave class defined in Section~\ref{sec:consequences}, whose members match the Gaussian moments through order $2K$ exactly.
The question asks whether the Gaussian can still be approximated by counterexamples under these constraints, or whether some positive signal fractions guarantee concavity in a relative neighborhood of the Gaussian.
Such local concavity could hold even if the answer to Problem~\ref{prob:uniform-noise} is negative.
Distinguishing counterexamples in the entire class from counterexamples near the Gaussian would therefore clarify which properties Gaussian smoothing can restore.

\subsection{Can the midpoint entropy comparison be strengthened to strictly positive curvature?}
For every $q>4/7$, the main argument makes the midpoint entropy smaller than the entropy of some unequally weighted sum.
Appendix~\ref{app:midpoint-curvature} proves the stronger statement $F''(1/2)>0$ only for $q>2-\sqrt2$.
An interval between these parameter ranges is still uncovered.
\begin{problem}\label{prob:midpoint-curvature}
For every $q\in(4/7,2-\sqrt2]$, does there exist a sequence of strictly positive, smooth, symmetric, strongly log-concave densities $f_n$ with mean zero and variance one,
converging to the standard Gaussian in each fixed-order relative $C_b^m$ norm, such that, for all sufficiently large $n$,
\[
 F_{\mathcal S_qf_n}''(1/2)>0
 \,?
\]
\end{problem}
Since the weighted-sum entropy is symmetric about $t=1/2$, strictly positive curvature would make equal weighting a strict local minimum.
This is finer than a two-point entropy comparison and does not follow directly from asymptotic entropy gaps at fixed weights.
It requires control of the differentiated error or a new construction allowing exact evaluation of midpoint curvature.

These questions ask, respectively, whether Gaussian smoothing restores concavity throughout the symmetric class, whether it restores concavity only near the Gaussian, and what local shape counterexamples have near equal weighting.
An answer to any of them would clarify the roles of symmetry and Gaussian smoothing in entropy concavity.

\clearpage
\appendix
\section*{Appendices}
The appendices address concrete questions arising from the two main theorems. Appendices~\ref{sec:endpoint}--\ref{sec:asymmetric} first explain why asymmetric counterexamples can be detected by the opposite signs of two moments, and supply the expansions and estimates used in the main text.
We then turn to symmetric constructions: Appendix~\ref{sec:main-family} gives explicit counterexamples based on a single cosine, while Appendix~\ref{sec:periodic-resonance} extends this phenomenon to general smooth even periodic profiles and computes a scaling limit near the midpoint.
This raises two further questions: is the midpoint a strict local minimum, and why does the method of the main text encounter the bound $4/7$? Appendices~\ref{app:midpoint-curvature} and~\ref{sec:interpolation} answer these questions through curvature calculations and sharp interpolation estimates, respectively.
Finally, Appendix~\ref{sec:finite-certificate} specifies finite parameter values and a strictly positive lower bound for the entropy gap, turning asymptotic existence into an explicit example.
\section{Endpoint expansions for asymmetric densities}\label{sec:endpoint}
The asymmetric construction requires a local criterion: which densities make the entropy of the weighted sum strictly convex near the endpoint?
We prove that the leading term of the endpoint curvature is determined by the product of the third central moment and the third moment of the logarithmic derivative, namely
$F''(t)=-\mu_3(f)J_3(f)/(16\sqrt t)+O(1)$.
Thus $\mu_3(f)J_3(f)<0$ suffices to contradict concavity, reducing the construction to adjusting the signs of two integrals.
The proof establishes four continuous derivatives of the entropy in the parameter $s=\sqrt t$ and then controls the remainder after the change of variables.

\begin{proposition}\label{prop:endpoint}
Let $f=\varphi(1+r)$, where $r$ is a Schwartz function, $\|r\|_\infty<1/2$, and $f$ is a probability density of mean zero and variance one. Write
\[
 \mu_3(f)=\int x^3f(x)\,dx,\qquad
 s_f=(\log f)',\qquad J_3(f)=\int f(x)s_f(x)^3\,dx.
\]
If $X,Y$ are independent with density $f$, then $F(t)=h(\sqrt{1-t}X+\sqrt tY)$ is twice differentiable on $0<t<1$, and
\begin{equation}\label{eq:endpoint}
 F''(t)=-\frac{\mu_3(f)J_3(f)}{16\sqrt t}+O(1)\qquad(t\downarrow0).
\end{equation}
Moreover, writing $I(f)=\int fs_f^2$, the function $F$ extends continuously to $[0,1]$, with $F(0)=F(1)=h(X)$, and
\begin{equation}\label{eq:entropy-expansion}
 F(t)=h(X)+\frac{I(f)-1}{2}t-\frac{\mu_3(f)J_3(f)}{12}t^{3/2}+O(t^2).
\end{equation}
The constants in both remainders may depend on the fixed density $f$.
\end{proposition}
\begin{proof}
Put $a(s)=\sqrt{1-s^2}$ and $Q(s)=h(a(s)X+sY)$ for $|s|<1$, and denote the density by $p_s$. In the two-dimensional integral make the orthogonal change of variables
\[
 x=a(s)z-sw,\qquad y=sz+a(s)w.
\]
Its Jacobian is one and $x^2+y^2=z^2+w^2$. Integrating the product density $f(x)f(y)$ with respect to $w$ gives
\begin{equation}\label{eq:rotated-density}
 p_s(z)=\varphi(z)R_s(z),\quad
 R_s(z)=\mathbb E\big[(1+r(a(s)z-sG))(1+r(sz+a(s)G))\big],
 \quad G\sim N(0,1).
\end{equation}
Each factor lies between $1/2$ and $3/2$, so $1/4<R_s(z)<9/4$. Hence $p_s\log p_s$ is integrable: $|\log R_s|$ is uniformly bounded, and $|\log\varphi(z)|$ grows at most quadratically.

Fix $0<s_0<1$. The first four derivatives of $a$ are bounded on $[-s_0,s_0]$. Differentiating the integrand in~\eqref{eq:rotated-density} $k\le4$ times, each chain-rule term is a product of bounded derivatives of $r$, at most $k$ linear factors in $z,G$, and bounded derivatives of $a$. Its absolute value is therefore at most $C_k(1+|z|+|G|)^k$. All Gaussian moments are finite, so dominated convergence permits continuous differentiation under the expectation and gives
\[
 |\partial_s^kR_s(z)|\le C_k(1+|z|)^k,\qquad 0\le k\le4.
\]
To apply these bounds to entropy, put $\Psi(u)=u\log u$. On $u\in[1/4,9/4]$, the derivatives
\[
 \Psi'(u)=1+\log u,\quad \Psi''(u)=u^{-1},\quad
 \Psi^{(3)}(u)=-u^{-2},\quad \Psi^{(4)}(u)=2u^{-3}
\]
are bounded. Write $p_s\log p_s=\varphi R_s\log\varphi+\varphi\Psi(R_s)$. For the second term, the fourth derivative involves
\begin{align*}
 \partial_s^4\Psi(R_s)
 &=\Psi'(R_s)R_s^{(4)}+4\Psi''(R_s)R_s'R_s^{(3)}
   +3\Psi''(R_s)(R_s'')^2\\
 &\quad+6\Psi^{(3)}(R_s)(R_s')^2R_s''
   +\Psi^{(4)}(R_s)(R_s')^4.
\end{align*}
Superscripts and primes on $R_s$ denote derivatives with respect to $s$, while derivatives of $\Psi$ are with respect to $u$. Each product has total derivative order four and thus grows at most quartically in $z$; the first term gains at most two further powers from $\log\varphi$. Lower derivatives are treated by the same product rule. Hence $\partial_s^k(p_s\log p_s)$ is dominated by
$C\varphi(z)(1+|z|)^6$, which is integrable, for $k\le4$. Another application of dominated convergence gives $Q\in C^4([-s_0,s_0])$. Write
\[
 M_4=\sup_{|s|\le s_0}|Q^{(4)}(s)|<\infty.
\]
Since $s_0$ is arbitrary below one, the chain rule also proves that $F$ is twice differentiable on $(0,1)$.

To compute the derivatives at $s=0$, alternatively write
\[
 p_s(z)=\mathbb E\left[a(s)^{-1}f\left(\frac{z-sY}{a(s)}\right)\right].
\]
Set $b(s)=a(s)^{-1}$ and $w(s)=b(s)(z-sy)$. At zero,
\[
 (b,b',b'',b''')(0)=(1,0,1,0),\qquad
 (w,w',w'',w''')(0)=(z,-y,z,-3y).
\]
The product and chain rules give
\begin{align*}
 \left.\partial_s\{bf(w)\}\right|_0&=-yf'(z),\\
 \left.\partial_s^2\{bf(w)\}\right|_0
 &=\left.b''f+2b'f'w'+b\{f''(w')^2+f'w''\}\right|_0\\
 &=f(z)+y^2f''(z)+zf'(z),\\
 \left.\partial_s^3\{bf(w)\}\right|_0
 &=b'''f+3b''f'w'+3b'\{f''(w')^2+f'w''\}\big|_0\\
 &\quad+b\{f'''(w')^3+3f''w'w''+f'w'''\}\big|_0\\
 &=-y^3f'''(z)-3zyf''(z)-6yf'(z).
\end{align*}
All derivatives of $f$ are bounded. For fixed $z$ and a compact $s$-interval, these derivatives and their continuity are controlled by $C_z(1+|Y|)^4$. All moments of $Y$ are finite since $f\le3\varphi/2$. We may therefore differentiate before taking the expectation in $Y$. Using $\mathbb EY=0$ gives
\[
 \partial_sp_s\big|_0=0,\quad \partial_s^2p_s\big|_0=f+zf'+f'',\quad \partial_s^3p_s\big|_0=-\mu_3(f)f'''(z).
\]
Since $p_s$ is a probability density, $\int p_s=1$. The preceding derivative bounds also apply to $p_s$ itself, so $\int\partial_s^kp_s=0$ for $1\le k\le4$. Differentiating twice first gives
\[
 Q''(s)=-\int(\partial_s^2p_s)(1+\log p_s)
          -\int\frac{(\partial_sp_s)^2}{p_s}.
\]
At $s=0$, the second term vanishes. Since $f+zf'=(zf)'$, integration by parts gives
\[
 \int(f+zf')\log f=-\int zf'=1,\qquad
 \int f''\log f=-\int\frac{(f')^2}{f}=-I(f).
\]
The boundary terms $zf\log f$, $zf$, and $f'\log f$ tend to zero: $f$ and its derivatives have Gaussian decay with at most polynomial factors, and $\log f$ grows at most quadratically. Also $s_f=-z+r'/(1+r)$ grows at most linearly, so $I(f)<\infty$. Consequently $Q''(0)=I(f)-1$.
Differentiating the entropy once more gives
\[
 Q'''(s)=-\int\left\{(\partial_s^3p_s)(1+\log p_s)
 +\frac{3(\partial_sp_s)(\partial_s^2p_s)}{p_s}
 -\frac{(\partial_sp_s)^3}{p_s^2}\right\}.
\]
The preceding domination justifies this identity. Since $\int\partial_s^3p_s=0$, at zero we obtain
\begin{equation}\label{eq:Qthird}
 Q'(0)=0,\qquad Q'''(0)=\mu_3(f)\int f'''\log f.
\end{equation}

To compute $\int f^{(3)}\log f$, use $f'=fs_f$ to get $f''=f(s_f'+s_f^2)$. Integrating by parts on $[-R,R]$ and then letting $R\to\infty$ gives
\[
 \int f'''\log f=-\int f''s_f
 =-\int fs_fs_f'-\int fs_f^3.
\]
Also, $(fs_f^2)'=fs_f^3+2fs_fs_f'$ yields
\[
 \int fs_fs_f'=-\tfrac12\int fs_f^3.
\]
Here $s_f=-x+r'/(1+r)$, $s_f'$ is bounded, and $\log f$ grows at most quadratically; $f$ and its derivatives have Gaussian decay. Thus $f''\log f$ and $fs_f^2$ vanish at both ends, and all integrands are absolutely integrable. Combining the two identities gives
\[
 \int f'''\log f=-J_3(f)/2,\qquad Q'''(0)=-\mu_3(f)J_3(f)/2.
\]

For $s>0$, $F(s^2)=Q(s)$, hence
\[
 F''(s^2)=\frac{sQ''(s)-Q'(s)}{4s^3}.
\]
Taylor's formula gives, respectively,
\begin{align*}
 Q'(s)&=Q''(0)s+\tfrac12Q'''(0)s^2+R_1(s),&|R_1(s)|&\le M_4s^3/6,\\
 Q''(s)&=Q''(0)+Q'''(0)s+R_2(s),&|R_2(s)|&\le M_4s^2/2.
\end{align*}
Upon subtraction, the linear terms cancel, and the cubic remainder satisfies
$|sR_2-R_1|\le2M_4s^3/3$. Therefore
\begin{equation}\label{eq:endpoint-error}
 \left|F''(s^2)-\frac{Q'''(0)}{8s}\right|\le M_4/6,
 \qquad 0<s<s_0.
\end{equation}
Substituting~\eqref{eq:Qthird} and the integral identity proves~\eqref{eq:endpoint}. Applying Taylor's formula to $Q$ itself gives
\[
 Q(s)=Q(0)+\frac{Q''(0)}2s^2+\frac{Q^{(3)}(0)}6s^3+R_0(s),
 \qquad |R_0(s)|\le M_4s^4/24.
\]
Taking $s=\sqrt t$ proves~\eqref{eq:entropy-expansion} and $F(t)\to h(X)$ as $t\downarrow0$. Since $X,Y$ are independent copies, exchanging them preserves the law of the weighted sum, so $F(1-t)=F(t)$. This also gives continuity and the endpoint value as $t\uparrow1$.
\end{proof}

For an even density, $\mu_3(f)=J_3(f)=0$, so the singular cubic contribution vanishes. Symmetric counterexamples therefore require a separate construction.

\section{Moment estimates for asymmetric perturbations}\label{sec:asymmetric}
Appendix~\ref{sec:endpoint} reduces the construction to $\mu_3(f)J_3(f)<0$. We explain how to control these factors independently.
The fourth and fifth Hermite components preserve the first three moments, while their even--odd cross term changes the third moment of the logarithmic derivative.
Adding a third-order component of smaller amplitude then adjusts the third central moment.
We give the required vanishing-moment identities, the exact cross integral, and a general cubic remainder bound.
These estimates make the sign choices in the main argument strict for every fixed $q>0$.
We use the following three probabilists' Hermite polynomials:
\[
 H_3(y)=y^3-3y,\quad H_4(y)=y^4-6y^2+3,\quad
 H_5(y)=y^5-10y^3+15y.
\]

Throughout this appendix, $u,v$ are defined by~\eqref{eq:u}--\eqref{eq:v}.
\subsection{Moment conditions and the third moment of the logarithmic derivative}\label{sec:construction}
The endpoint formula expresses the required sign condition as $\mu_3(f)J_3(f)<0$. To control the factors separately, consider $r=\varepsilon u+\varepsilon^2v$ and $f_\varepsilon=\varphi(1+r)$. After the substitution $y=\sqrt2x$, the fourth and fifth Hermite polynomials in $u$ are orthogonal to every polynomial of degree at most three. Thus $u$ preserves normalization and the first three moments. The third Hermite polynomial in $v$ preserves normalization, mean, and variance, while increasing the third central moment by $3\varepsilon^2/2$.

The cross term between the fourth and fifth components determines the quadratic coefficient of the third logarithmic-derivative moment. Identity~\eqref{eq:score-exact} below shows that this coefficient in $J_3(f_\varepsilon)$ depends on $K=\int x(u')^2\varphi$. If $u$ has a single parity, $(u')^2$ is even and the integral vanishes. For our combination the cross term gives $K=560\sqrt6/27>1$, making the coefficient $3(1-K)$ negative. The following lemmas derive these moments and identities, before controlling the higher-order remainder.

\begin{lemma}\label{lem:moments}
For the functions in~\eqref{eq:u} and~\eqref{eq:v},
\[
 \int x^ju\varphi=0\quad(0\le j\le3),\qquad
 \int x^jv\varphi=0\quad(0\le j\le2),\qquad
 \int x^3v\varphi=3/2.
\]
\end{lemma}
\begin{proof}
The probabilists' Hermite polynomials satisfy Rodrigues' identity
$H_k(y)\varphi(y)=(-1)^k\varphi^{(k)}(y)$. Integrating a polynomial $p$ by parts $k$ times, every boundary term vanishes because it is a polynomial times a derivative of a Gaussian density. Hence
\[
 \int p(y)H_k(y)\varphi(y)\,dy=\int p^{(k)}(y)\varphi(y)\,dy.
\]
If $\deg p<k$, the right-hand side is zero; if $p(y)=y^k$, it equals $k!$.

Set $y=\sqrt2x$. Since $\varphi(x)e^{-x^2/2}\,dx=\varphi(y)\,dy/\sqrt2$,
\[
 \int x^je^{-x^2/2}H_k(\sqrt2x)\varphi(x)\,dx
 =2^{-(j+1)/2}\int y^jH_k(y)\varphi(y)\,dy.
\]
Taking $k=4,5$ and $j\le3$ for $u$, and $k=3$ and $j\le2$ for $v$, proves all vanishing moments. Finally, $j=k=3$ gives $2^{-2}3!=3/2$.
\end{proof}

\begin{lemma}\label{lem:score}
Let $r=\varepsilon u+\varepsilon^2v$, assume $\|r\|_\infty<1/2$, and put $f_\varepsilon=\varphi(1+r)$. Then $f_\varepsilon$ is a probability density of mean zero and variance one, and
\begin{align}
 \mu_3(f_\varepsilon)&=3\varepsilon^2/2,\label{eq:mu3}\\
 J_3(f_\varepsilon)&=2\mu_3(f_\varepsilon)-3\int\frac{x(r')^2}{1+r}\varphi
                  +\int\frac{(r')^3}{(1+r)^2}\varphi.\label{eq:score-exact}
\end{align}
Moreover,
\begin{equation}\label{eq:K}
 K:=\int x(u')^2\varphi=\frac{560\sqrt6}{27}.
\end{equation}
\end{lemma}
\begin{proof}
Positivity follows from $1+r>1/2$. Lemma~\ref{lem:moments} and the mean and variance of the standard Gaussian give $\int f_\varepsilon=1$, $\int xf_\varepsilon=0$, and $\int x^2f_\varepsilon=1$; the same lemma gives~\eqref{eq:mu3}.

Using $s_{f_\varepsilon}=-x+r'/(1+r)$, expand $(1+r)s_{f_\varepsilon}^3$ to obtain
\[
 (1+r)s_{f_\varepsilon}^3=-x^3(1+r)+3x^2r'
 -\frac{3x(r')^2}{1+r}+\frac{(r')^3}{(1+r)^2}.
\]
Since $\varphi'=-x\varphi$, integration by parts gives
\[
 \int x^2r'\varphi=-\int r(2x-x^3)\varphi=\mu_3(f_\varepsilon).
\]
The last step uses $\int xr\varphi=0$ and $\int x^3\varphi=0$. The first two integrals therefore sum to $-\mu_3(f_\varepsilon)+3\mu_3(f_\varepsilon)=2\mu_3(f_\varepsilon)$, proving~\eqref{eq:score-exact}.

To compute $K$, write
\[
 u=e^{-x^2/2}(E+\sqrt2O),\quad
 E=4x^4-12x^2+3,\quad O=4x^5-20x^3+15x.
\]
Then
\[
 E'-xE=-4x^5+28x^3-27x,
 \qquad O'-xO=-4x^6+40x^4-75x^2+15.
\]
The first polynomial is odd and the second even. After multiplication by $x$, both square terms are odd and integrate to zero, leaving only the cross term. Direct multiplication gives
\[
 2x(E'-xE)(O'-xO)
 =32x^{12}-544x^{10}+3056x^8-6480x^6+4890x^4-810x^2.
\]
Let $Z\sim N(0,1/3)$. Then
$\int x^{2j}e^{-3x^2/2}\,dx/\sqrt{2\pi}=\mathbb EZ^{2j}/\sqrt3$.
Gaussian integration by parts gives the even-moment recurrence
$\mathbb EG^{2j}=(2j-1)\mathbb EG^{2j-2}$, starting from $\mathbb E1=1$. Thus
$\mathbb EZ^{2j}=(2j-1)!!/3^j$, and hence
\begin{align*}
 K&=\sqrt{2/3}\left(
 32\frac{10395}{729}-544\frac{945}{243}
 +3056\frac{105}{81}-6480\frac{15}{27}
 +4890\frac39-810\frac13\right)\\
 &=\sqrt{2/3}\,\frac{560}{9}
 =\frac{560\sqrt6}{27}.
\end{align*}
\end{proof}

\subsection{A bound for the third moment of the logarithmic derivative}
To determine the sign after Gaussian smoothing, we extend the expansion to any two perturbations satisfying the vanishing-moment conditions and give a uniform cubic remainder bound.
\begin{lemma}\label{lem:general-score-remainder}
Let $U,V$ be real Schwartz functions satisfying
\[
 \int x^jU\varphi=0\quad(0\le j\le3),\qquad
 \int x^jV\varphi=0\quad(0\le j\le2).
\]
Let $M\ge\max\{1,\|U\|_{C_b^1},\|V\|_{C_b^1}\}$, and take
$0<\varepsilon\le1$, $2\varepsilon M\le1/2$.
The density $g=\varphi(1+r)$, with $r=\varepsilon U+\varepsilon^2V$, satisfies
\[
 \left|J_3(g)-\varepsilon^2\{2m_3(V)-3K_U\}\right|
 \le62M^3\varepsilon^3,
\]
where $m_3(V)=\int x^3V\varphi$ and $K_U=\int x(U')^2\varphi$.
\end{lemma}
\begin{proof}
The moment conditions give normalization, mean zero, variance one, and $\mu_3(g)=\varepsilon^2m_3(V)$.
Expanding $(-x+r'/(1+r))^3$ and integrating $\int x^2r'\varphi$ by parts, as in~\eqref{eq:score-exact}, gives
\[
 J_3(g)=2\varepsilon^2m_3(V)-3\int\frac{x(r')^2}{1+r}\varphi
             +\int\frac{(r')^3}{(1+r)^2}\varphi.
\]
Schwartz decay makes all boundary terms vanish. The assumptions imply $|r|,|r'|\le2\varepsilon M$, so
\[
 \frac1{1+r}\le2,\quad \frac1{(1+r)^2}\le4,
 \quad\left|\frac1{1+r}-1\right|\le2|r|\le4\varepsilon M.
\]
Also, $r'=\varepsilon U'+\varepsilon^2V'$ gives
\[
 |(r')^2-\varepsilon^2(U')^2|
 \le2\varepsilon^3M^2+\varepsilon^4M^2\le3\varepsilon^3M^2.
\]
Write the difference as
\[
 \frac{(r')^2}{1+r}-\varepsilon^2(U')^2
 =\frac{(r')^2-\varepsilon^2(U')^2}{1+r}
 +\varepsilon^2(U')^2\left(\frac1{1+r}-1\right),
\]
Its absolute value is at most $(6M^2+4M^3)\varepsilon^3$. By Cauchy--Schwarz, $\int|x|\varphi\le(\int x^2\varphi)^{1/2}=1$, so the error in the second term of~\eqref{eq:score-exact} relative to $-3K_U\varepsilon^2$ is at most
$(18M^2+12M^3)\varepsilon^3$. The absolute value of the third term is at most
$4(2\varepsilon M)^3=32M^3\varepsilon^3$. The first term is exactly $2\varepsilon^2m_3(V)$. Adding these bounds and using $M\ge1$ proves the claim.
\end{proof}

For the fixed functions $u,v$ above, $m_3(v)=3/2$, so the lemma gives
$J_3(\varphi(1+\varepsilon u+\varepsilon^2v))
=3(1-K)\varepsilon^2+O(\varepsilon^3)$.
Since $K>1$, sufficiently small positive amplitudes give asymmetric endpoint counterexamples. Section~\ref{sec:asymmetric-noise} realizes the same sign condition after Gaussian smoothing.

\section{An explicit family of cosine counterexamples}\label{sec:main-family}
The Hermite construction accommodates both noise and moment constraints, but without added noise a simpler counterexample is available.
Adding a single cosine mode to the Gaussian density, we prove that for every $k\ge14$ the entropy gap between weights $1/5$ and $1/2$ is strictly greater than $e^{-4k}/256$.
Equal weighting retains a quadratic oscillatory term, whereas the chosen unequal weights exponentially damp all oscillatory terms.
This explicit example also motivates the next appendix: does the phenomenon depend on the particular choice of a cosine waveform?
We do not normalize the variance here, since common scaling does not affect concavity.
\begin{proposition}\label{thm:family}
For every real $k\ge14$, let
\begin{equation}\label{eq:family}
 f_k(x)=\frac{\varphi(x)(1+e^{-k}\cos(kx))}{1+e^{-k-k^2/2}}.
\end{equation}
Then $f_k$ is a strictly positive, even, $C^\infty$ probability density satisfying
$(\log f_k)''<-1/2$, and
\begin{equation}\label{eq:family-gap}
 F_{f_k}(1/5)-F_{f_k}(1/2)>\frac{e^{-4k}}{256}.
\end{equation}
In particular, the midpoint is not a global maximizer of $F_{f_k}$, and $F_{f_k}$ is not concave.
Moreover, for every fixed nonnegative integer $j$,
\begin{equation}\label{eq:family-convergence}
 \left\|\frac{d^j}{dx^j}\left(\frac{f_k}{\varphi}-1\right)\right\|_\infty\longrightarrow0,
 \qquad
 \|(\log f_k)''+1\|_\infty\longrightarrow0
 \quad(k\longrightarrow\infty).
\end{equation}
\end{proposition}

The threshold $14$ is a sufficient condition for the uniform estimates below. The amplitude in~\eqref{eq:family} tends to zero as the frequency increases;
the derivative order in~\eqref{eq:family-convergence} is always fixed.

The proof distinguishes the cross terms under equal and unequal weighting. With equal weights, multiplying two cosine factors produces a frequency
that is not damped by Gaussian integration. At $t=1/5$, every oscillatory term is exponentially damped.
Since the two weighted sums have equal variance, comparing their entropies is equivalent to reversing the comparison of their relative entropies with respect to the standard Gaussian.
We first treat amplitude and frequency as independent parameters and then set the amplitude to $e^{-k}$.

\subsection{Convolution formulas and entropy comparison}\label{sec:family-proof}
Fix $k\ge14$ and $0<\eta\le10^{-3}$, and put
\[
 C=1+\eta e^{-k^2/2},\qquad
 f_{\eta,k}(x)=C^{-1}\varphi(x)(1+\eta\cos(kx)).
\]
The Gaussian integral $\int\varphi(x)e^{iux}\,dx=e^{-u^2/2}$ shows that $C$ is exactly the normalization constant.
Since $1-\eta>0$, $f_{\eta,k}$ is a positive, even, smooth density bounded above and below by positive multiples of $\varphi$.
It therefore has mean zero and finite absolute moments of all orders. The same Gaussian control will ensure finiteness of the entropies and integrals below.

\subsubsection{Weighted convolution}
Let $a=\sqrt{1-t}$ and $b=\sqrt t$ for $t\in[0,1]$. For independent copies
$X,Y$ with density $f_{\eta,k}$, denote the density of $aX+bY$ by $p_t$. Make the orthogonal change of variables with determinant one
\[
 x=az-bw,\qquad y=bz+aw.
\]
Since $x^2+y^2=z^2+w^2$, integration in $w$ gives $p_t(z)=\varphi(z)R_t(z)$, where
\begin{equation}\label{eq:rotation}
 R_t(z)=C^{-2}\E\big[(1+\eta\cos(k(az-bG)))
                         (1+\eta\cos(k(bz+aG)))\big],\quad G\sim N(0,1).
\end{equation}
Formula~\eqref{eq:rotation} also holds at $t=0,1$. Using
$2\cos u\cos v=\cos(u+v)+\cos(u-v)$ and
$\E\cos(c+dG)=e^{-d^2/2}\cos c$ gives the exact formula
\begin{align}
 C^2R_t(z)&=1+\eta e^{-k^2b^2/2}\cos(kaz)
                 +\eta e^{-k^2a^2/2}\cos(kbz)\notag\\
 &\quad+\frac{\eta^2}{2}e^{-k^2(a-b)^2/2}\cos(k(a+b)z)\notag\\
 &\quad+\frac{\eta^2}{2}e^{-k^2(a+b)^2/2}\cos(k(a-b)z).
 \label{eq:convolution}
\end{align}
Equation~\eqref{eq:rotation} also immediately implies
\begin{equation}\label{eq:ratio-bounds}
 0<\frac{(1-\eta)^2}{C^2}\le R_t(z)
 \le\frac{(1+\eta)^2}{C^2}<2.
\end{equation}
These bounds are uniform in $t,z$. Hence $|p_t\log p_t|$ is bounded by $A\varphi(z)(1+z^2)$,
where $A$ depends only on $\eta,k$. Entropy is therefore finite. Since~\eqref{eq:convolution} is continuous in $t$,
dominated convergence gives $F_{f_{\eta,k}}\in C([0,1])$.

\subsubsection{Quadratic entropy bounds}
For positive densities $p,q$, write $D(p\Vert q)=\int p\log(p/q)$. In the present notation, put
\[
 D_t=D(p_t\Vert\varphi)=\int\varphi R_t\log R_t,
 \qquad \chi_t^2=\int\varphi(R_t-1)^2.
\]
Apply the integral Taylor formula from the proof of Lemma~\ref{lem:entropy-quadratic}.
Since $0<R_t<2$, the second derivative of $u\mapsto u\log u-u+1$ is at least $1/2$ along the segment joining $1$ to $R_t$.
On the other hand, $\log u\le u-1$. Together with $\int\varphi(R_t-1)=0$, this yields
\begin{equation}\label{eq:quadratic}
 \tfrac14\chi_t^2\le D_t\le\chi_t^2.
\end{equation}

Let $\sigma^2=\E X^2$. Independence and zero means give
$\E(aX+bY)^2=(a^2+b^2)\sigma^2=\sigma^2$. Thus
\begin{equation}\label{eq:entropy-relative}
 F_{f_{\eta,k}}(t)=\frac12\log(2\pi)+\frac12\sigma^2-D_t.
\end{equation}
In particular, even when $\sigma^2\ne1$, we have
$F_{f_{\eta,k}}(1/5)-F_{f_{\eta,k}}(1/2)=D_{1/2}-D_{1/5}$.

\subsubsection{Equal and unequal weights}
At $t=1/2$, we have $a=b=1/\sqrt2$, so the first quadratic cross term in~\eqref{eq:convolution}
is not exponentially damped. Centering the formula using $\E R_{1/2}(G)=1$ gives
\begin{align*}
 R_{1/2}(G)-1=C^{-2}\big[&2\eta e^{-k^2/4}
       \{\cos(kG/\sqrt2)-e^{-k^2/4}\}\\
       &+\tfrac12\eta^2\{\cos(\sqrt2kG)-e^{-k^2}\}\big].
\end{align*}
For any nonnegative $u,v$, the cosine product identity gives
\begin{align*}
 \operatorname{Cov}(\cos(uG),\cos(vG))
 &=\tfrac12 e^{-(u-v)^2/2}+\tfrac12 e^{-(u+v)^2/2}
       -e^{-(u^2+v^2)/2}\\
 &=e^{-(u^2+v^2)/2}(\cosh(uv)-1)\ge0.
\end{align*}
The coefficients of the two centered terms are positive. Retaining the variance of the second term therefore gives
\[
 \chi_{1/2}^2\ge\frac{\eta^4}{4C^4}
       \operatorname{Var}(\cos(\sqrt2kG))
 =\frac{\eta^4}{8C^4}(1-e^{-2k^2})^2.
\]
Since $C^4\le(1001/1000)^4<2$ and $e^{-2k^2}<1/4$, we have
$\operatorname{Var}(\cos(\sqrt2kG))>9/32>1/4$. Hence
\begin{equation}\label{eq:equal-bound}
 \chi_{1/2}^2>\frac{\eta^4}{32},\qquad
 D_{1/2}>\frac{\eta^4}{128}.
\end{equation}

At $t=1/5$, $a=2/\sqrt5$ and $b=1/\sqrt5$. The four oscillatory coefficients have damping factors
\[
 e^{-k^2/10},\quad e^{-2k^2/5},\quad e^{-k^2/10},\quad e^{-9k^2/10}.
\]
Since $C\ge1$ and $C^2-1=2\eta e^{-k^2/2}+\eta^2e^{-k^2}$,
formula~\eqref{eq:convolution} yields
\begin{align*}
 |R_{1/5}(z)-1|
 &\le\eta(e^{-k^2/10}+e^{-2k^2/5})
      +\frac{\eta^2}{2}(e^{-k^2/10}+e^{-9k^2/10})\\
 &\quad+2\eta e^{-k^2/2}+\eta^2e^{-k^2}\\
 &\le(4\eta+2\eta^2)e^{-k^2/10}
 \le6\eta e^{-k^2/10}.
\end{align*}
Consequently,
\begin{equation}\label{eq:unequal-bound}
 D_{1/5}\le36\eta^2e^{-k^2/5}.
\end{equation}
Equations~\eqref{eq:equal-bound} and~\eqref{eq:unequal-bound} are the two entropy bounds we need.

\subsubsection{Proof for the counterexample family}
\begin{proof}[Proof of Proposition~\ref{thm:family}]
Take $\eta=e^{-k}$. Since $e^{14}>2^{14}>1000$, the amplitude restriction holds.
Logarithmic differentiation gives
\begin{align}
 (\log f_k)''(x)
 &=-1-\frac{\eta k^2\cos(kx)}{1+\eta\cos(kx)}
       -\frac{\eta^2k^2\sin^2(kx)}{(1+\eta\cos(kx))^2}\notag\\
 &\le-1+\frac{k^2e^{-k}}{1-e^{-k}}.
 \label{eq:curvature}
\end{align}
The function $k^2e^{-k}$ decreases for $k\ge14$, while $1-e^{-k}$ increases. Hence
\[
 \frac{k^2e^{-k}}{1-e^{-k}}
 \le\frac{196e^{-14}}{1-e^{-14}}
 <\frac{196}{2^{14}-1}<\frac13.
\]
Thus $(\log f_k)''<-2/3<-1/2$.

The function $k^2/5-2k$ increases for $k\ge14$, with initial value $56/5>10$.
The exponential series gives $e>1+1+1/2+1/6=8/3$, and
$8^{10}>9216\cdot3^{10}$. Therefore
\[
 e^{k^2/5-2k}>e^{10}>(8/3)^{10}>9216=36\cdot256.
\]
Together with~\eqref{eq:unequal-bound}, this gives $D_{1/5}<e^{-4k}/256$;
then~\eqref{eq:equal-bound} and~\eqref{eq:entropy-relative} give~\eqref{eq:family-gap}.
Exchanging $X,Y$ gives $F_{f_k}(1-t)=F_{f_k}(t)$, and hence
\[
 F_{f_k}(1/2)<\frac{F_{f_k}(1/5)+F_{f_k}(4/5)}2,
\]
This directly violates concavity.

Finally, put $q_k=f_k/\varphi$. The explicit formula gives
\[
 \|q_k-1\|_\infty\le2e^{-k},\qquad
 \|q_k^{(j)}\|_\infty\le k^je^{-k}\quad(j\ge1).
\]
For each fixed $j$, the right-hand side tends to zero. The identity in~\eqref{eq:curvature} also gives
\[
 \|(\log f_k)''+1\|_\infty
 \le\frac{k^2e^{-k}}{1-e^{-k}}
       +\frac{k^2e^{-2k}}{(1-e^{-k})^2}\longrightarrow0.
\]
This proves~\eqref{eq:family-convergence}.
\end{proof}

\section{Periodic perturbations and a scaling limit near the midpoint}\label{sec:periodic-resonance}
The single-cosine example is not isolated among waveforms. We prove that every nonconstant smooth even periodic waveform, with a suitable small amplitude and high frequency, eventually gives entropy greater than the midpoint entropy at every fixed interior nonmidpoint weight.
At fixed weights the entropy deficit is determined by matching integer frequencies, whereas the scale $t=1/2+s/k$ yields a continuous limiting profile.
This extension describes both the mechanism of the counterexamples and the scale on which the entropy gap concentrates near the midpoint.
We first give explicit bounds for an exponential cosine potential, then treat general waveforms and the scaling limit.
\subsection{An exponential cosine potential}
The exponential form ensures positivity automatically and reduces log-concavity to an estimate of the second derivative of the potential. With amplitude $k^{-3}$, it gives an entropy gap greater than $1/(25k^{12})$ for every $k\ge10$.
\begin{proposition}\label{prop:exp-cos-simple}
For every real $k\ge10$, set $\epsilon=k^{-3}$ and define
\[
 p_k(y)=\frac{\varphi(y)e^{\epsilon\cos(ky)}}{Z_k},\qquad
 Z_k=\int\varphi(y)e^{\epsilon\cos(ky)}\,dy.
\]
Let $Y_k$ have this density, and put $\sigma_k^2=\mathbb EY_k^2$ and $X_k=Y_k/\sigma_k$.
Then $X_k$ is symmetric, has variance one, and has a positive analytic strongly log-concave density. If $X_k'$ is an independent copy and
$F_k(t)=h(\sqrt tX_k+\sqrt{1-t}X_k')$, then
\[
 F_k(1/5)=F_k(4/5)>F_k(1/2)+\frac1{25k^{12}}.
\]
Moreover, the negative logarithmic curvature of $p_k$ is exactly
\[
 - (\log p_k)''(y)=1+k^{-1}\cos(ky),
\]
After standardization, the negative logarithmic curvature still converges uniformly to $1$, and the density and its first two derivatives converge uniformly to the corresponding Gaussian derivatives.
\end{proposition}
\begin{proof}
Write $e^{\epsilon\cos\theta}=\sum_{m\in\mathbb Z}b_m e^{im\theta}$.
The exponential power series and $\cos\theta=(e^{i\theta}+e^{-i\theta})/2$ give
\[
 b_m=b_{-m}\ge0,\quad b_0\ge1,\quad b_1\ge\epsilon/2,
 \quad\sum_m b_m=e^\epsilon,\quad\sum_m m^2b_m=\epsilon e^\epsilon.
\]
Hence $1\le Z_k\le e^\epsilon$. Put $a=\sqrt t$ and $b=\sqrt{1-t}$.
The relative density with respect to $\varphi$ of the original weighted sum $aY_k+bY_k'$ is
\begin{equation}\label{eq:periodic-exact-output}
 R_t(z)=Z_k^{-2}\sum_{m,l\in\mathbb Z}b_m b_l
 e^{-k^2(-mb+la)^2/2}e^{ik(ma+lb)z}.
\end{equation}
This follows from the orthogonal decomposition of independent standard Gaussians $G,H$ given $aG+bH=z$,
namely $G=az-bW,H=bz+aW$. All Fourier series converge absolutely.
Put $D_t=\int\varphi R_t\log R_t$. Since all weights give the same variance,
$F_k(t)-F_k(1/2)=D_{1/2}-D_t$; standardization subtracts the same constant $\log\sigma_k$ from every entropy.
Also $e^{-4\epsilon}\le R_t\le e^{4\epsilon}$. Set
$\eta=e^{4\epsilon}-1<1$. Lemma~\ref{lem:entropy-quadratic} gives
\begin{equation}\label{eq:periodic-kl-bounds}
 \frac{\|R_t-1\|_{L^2(\varphi)}^2}{2(1+\eta)}
 \le D_t\le
 \frac{\|R_t-1\|_{L^2(\varphi)}^2}{2(1-\eta)}.
\end{equation}
At $t=1/2$, the modes $(m,l)=(1,1),(-1,-1)$ in~\eqref{eq:periodic-exact-output}
give a $\cos(\sqrt2kz)$ term with coefficient at least $\epsilon^2/(2e^{2\epsilon})$.
All remaining cosine coefficients are nonnegative, and for any real $u,v$,
\[
 \operatorname{Cov}(\cos(uG),\cos(vG))
 =e^{-(u^2+v^2)/2}(\cosh(uv)-1)\ge0.
\]
Retaining only this term therefore gives
\[
 D_{1/2}\ge\frac{\epsilon^4e^{-8\epsilon}}{16}(1-e^{-2k^2})^2.
\]
At $t=1/5$, every pair with $0<|m|+|l|\le2$ satisfies
$|-2m+l|\ge1$, so its damping factor is at most $e^{-k^2/10}$.
The sum of coefficients of total degree $|m|+|l|\ge3$ is at most
\[
 \sum_{j\ge3}\frac{(2\epsilon)^j}{j!}
 \le\frac43\epsilon^3e^{2\epsilon}.
\]
Indeed, every two-dimensional Fourier coefficient is nonnegative, and modes of total degree greater than two cannot arise from
the zeroth, first, or second Taylor terms of $e^{\epsilon(\cos\theta+\cos\psi)}$.
Subtracting Gaussian means term by term and using $|\cos x-\mathbb E\cos(uG)|\le2$ gives
\[
 \|R_{1/5}-1\|_\infty
 \le2\left((e^{2\epsilon}-1)e^{-k^2/10}
                  +\frac43\epsilon^3e^{2\epsilon}\right).
\]
For $k\ge10$, $e^{-k^2/10}\le\epsilon/20$: the function
$k^2/10-3\log k$ is increasing on this interval, and
$e^{10}>\sum_{j=0}^{14}10^j/j!=34374418261/1702701>20000$.
Thus
\[
 \|R_{1/5}-1\|_\infty
 \le\epsilon^2 e^{2\epsilon}(1/5+8\epsilon/3).
\]
Combining~\eqref{eq:periodic-kl-bounds} with $\epsilon\le1/1000$ gives
\[
 \frac{D_{1/2}-D_{1/5}}{\epsilon^4}
 \ge\frac{e^{-8/1000}(1-e^{-200})^2}{16}
 -\frac{e^{4/1000}(1/5+8/3000)^2}{2(2-e^{4/1000})}
 >\frac1{25}.
\]
For the last inequality, use $e^{-8/1000}>124/125$,
$e^{-200}<1/1000$, and $e^{4/1000}<250/249$. The left-hand side is therefore strictly greater than
\[
 \frac{124}{125}\frac{(999/1000)^2}{16}
 -\frac{250}{496}\left(\frac{76}{375}\right)^2
 =\frac{5743710649}{139500000000}>\frac1{25}.
\]
The bound $e^{-200}<1/1000$ follows from $e^{200}>1+200+200^2/2$.

Finally, the second derivative of the Gaussian characteristic function gives
\[
 \sigma_k^2=1-\frac{k^2}{Z_k}\sum_m m^2b_m e^{-m^2k^2/2},\qquad
 1-\epsilon k^2e^{\epsilon-k^2/2}\le\sigma_k^2<1.
\]
The original negative logarithmic curvature lies in $[1-1/k,1+1/k]$ and is multiplied by $\sigma_k^2$ upon standardization.
Uniform convergence of the density and its first two derivatives follows directly from $\epsilon k^2\to0$, Gaussian decay, and
$\sigma_k\to1$. This amplitude choice is used only for the stated second-order convergence.
\end{proof}

\subsection{Pointwise entropy deficits for arbitrary smooth even periodic waveforms}
For a general waveform, the leading entropy deficit is determined by frequency pairs that survive Gaussian integration. The following formula distinguishes rational and irrational weight ratios and gives a strict comparison with the midpoint at every fixed interior nonmidpoint weight.
\begin{proposition}\label{prop:periodic-wave-law}
Let $u$ be a nonconstant real even $C^\infty$, $2\pi$-periodic function with zero periodic mean, and write
\[
 u(\theta)=\sum_{m\ne0}c_m e^{im\theta},\qquad c_{-m}=c_m\in\mathbb R.
\]
Take $\epsilon_k>0$ satisfying
$\epsilon_k k^2\to0$ and $\log(1/\epsilon_k)=o(k^2)$.
Let $Y_k$ have density proportional to $\varphi(y)e^{\epsilon_k u(ky)}$, and set
$X_k=Y_k/\sqrt{\mathbb EY_k^2}$. For sufficiently large $k$, $X_k$ is symmetric, strongly log-concave, and of variance one.
Again let $F_k$ denote the entropy of the weighted sum of independent copies. For every fixed $t\in(0,1)$,
\begin{equation}\label{eq:periodic-deficit-limit}
 \lim_{k\to\infty}\frac{h(\varphi)-F_k(t)}{\epsilon_k^4}
 =\begin{cases}
 \displaystyle\frac12\sum_{j\ne0}|c_{pj}c_{lj}|^2,
 &\sqrt{t/(1-t)}=p/l,\ (p,l)=1,\ p,l\in\mathbb N,\\[3pt]
 0,&\sqrt{t/(1-t)}\notin\mathbb Q.
 \end{cases}
\end{equation}
In particular, for every fixed $t\in(0,1)\setminus\{1/2\}$ and sufficiently large $k$,
$F_k(t)>F_k(1/2)$. The threshold may depend on $t$; this is not a simultaneous assertion for all nonmidpoint weights.
\end{proposition}
\begin{proof}
Normalize the perturbation by its periodic mean:
\[
 w_\epsilon(\theta)=\frac{e^{\epsilon u(\theta)}}{
 (2\pi)^{-1}\int_0^{2\pi}e^{\epsilon u(s)}ds}
 =1+\epsilon v_\epsilon(\theta),\qquad
 v_\epsilon=\sum_{m\ne0}d_m(\epsilon)e^{im\theta}.
\]
For every fixed integer $r\ge0$, Taylor expansion in smooth-function algebras and integration by parts give
\[
 \sum_{m\ne0}(1+|m|)^r|d_m(\epsilon)-c_m|=O_r(\epsilon),
 \qquad \sum_{m\ne0}(1+|m|)^r|d_m(\epsilon)|=O_r(1).
\]
More explicitly, first expand $w_\epsilon=1+\epsilon u+O(\epsilon^2)$ in the $C^{r+2}$ norm,
and then integrate the Fourier coefficients by parts $r+2$ times.
The actual normalization constant relative to the Gaussian is
$C_k=\int\varphi(y)w_{\epsilon_k}(ky)dy=1+O(\epsilon_k e^{-k^2/2})$,
and $\sigma_k^2=1+O(\epsilon_k k^2e^{-k^2/2})$.
Put $a=\sqrt t,b=\sqrt{1-t}$. The original weighted sum has relative Gaussian density exactly
\[
 R_{k,t}=C_k^{-2}(1+\epsilon_k L_{k,t}+\epsilon_k^2M_{k,t}),
\]
where
\[
 \|L_{k,t}\|_\infty\le C_t e^{-k^2\min(t,1-t)/2},
\]
\[
 M_{k,t}(z)=\sum_{m,l\ne0}d_m(\epsilon_k)d_l(\epsilon_k)
 e^{-k^2(-mb+la)^2/2}e^{ik(ma+lb)z}.
\]
Since $\log(1/\epsilon_k)=o(k^2)$, the linear term and normalization error are negligible; more precisely,
\[
 \|(R_{k,t}-1)/\epsilon_k^2-M_{k,t}\|_\infty\longrightarrow0.
\]
The double Fourier coefficients are uniformly absolutely summable. If $a/b$ is irrational, every nonzero integer pair
has $-mb+la\ne0$, so dominated convergence gives $\|M_{k,t}\|_\infty\to0$.
If $a/b=p/l$, the undamped modes are exactly $(m,n)=(pj,lj)$, and hence
\[
 \left\|M_{k,t}-\sum_{j\ne0}c_{pj}c_{lj}
             e^{ikj\sqrt{p^2+l^2}z}\right\|_\infty\longrightarrow0.
\]
Here $n$ is the second mode index, distinguished from the denominator integer $l$.
The cross terms between distinct surviving frequencies decay exponentially in the Gaussian $L^2$ inner product. Thus
\[
 \|M_{k,t}\|_{L^2(\varphi)}^2\longrightarrow
 \sum_{j\ne0}|c_{pj}c_{lj}|^2.
\]
Again, absolute summability permits truncation to finitely many modes followed by uniform control of the tail.
Since $\|R_{k,t}-1\|_\infty=O_t(\epsilon_k^2)$, the second-order entropy expansion gives
\[
 D(\varphi R_{k,t}\Vert\varphi)/\epsilon_k^4
 \longrightarrow\frac12\lim\|M_{k,t}\|_{L^2(\varphi)}^2.
\]
The standardized entropy deficit minus this relative entropy is exactly
$\tfrac12(\log\sigma_k^2-\sigma_k^2+1)
 =O(\epsilon_k^2 k^4e^{-k^2})=o(\epsilon_k^4)$, proving the limit.

At the midpoint, the limit is $S/2$, where $S=\sum_{j\ne0}|c_j|^4>0$.
For $p\ne l$, Cauchy--Schwarz gives
\[
 \sum_{j\ne0}|c_{pj}c_{lj}|^2
 \le\left(\sum_{j\ne0}|c_{pj}|^4\right)^{1/2}
       \left(\sum_{j\ne0}|c_{lj}|^4\right)^{1/2}\le S.
\]
Equality is impossible. Indeed, assume without loss of generality that $p>l$, and choose the smallest positive index $m_0$ with $c_{m_0}\ne0$.
If $p$ does not divide $m_0$, the second inequality is already strict. If it does and equality holds throughout,
both sequences have squared norm $S$, and equality in Cauchy--Schwarz forces
$|c_{pj}|^2=|c_{lj}|^2$ for every $j$.
Taking $j=m_0/p$ gives a nonzero mode $lm_0/p<m_0$, a contradiction.
The irrational case has limit zero, likewise strictly smaller than $S/2$.

Finally, the original negative logarithmic curvature is $1-\epsilon_k k^2u''(ky)$ and converges uniformly to one.
Multiplication by $\sigma_k^2$ gives the standardized curvature. Symmetry gives mean zero.
All moments of fixed order converge to the Gaussian moments. If also $\epsilon_k k^r\to0$, the ordinary density and its first $r$
derivatives converge uniformly to their Gaussian counterparts; for example, $\epsilon_k=e^{-k}$ meets this condition for every fixed order simultaneously.
Standardization may change the relative Gaussian ratio in the tails, so this assertion uses derivative norms of the ordinary densities.
\end{proof}

\subsection{The scale of the dip near the midpoint}
The Fourier argument also gives a uniform conclusion without interchanging second derivatives. For any fixed $S_0<\infty$,
\begin{equation}\label{eq:periodic-central-profile}
 \sup_{|s|\le S_0}\left|
 \frac{h(\varphi)-F_k(1/2+s/k)}{\epsilon_k^4}
 -\frac12\sum_{m\ne0}|c_m|^4e^{-2m^2s^2}\right|\longrightarrow0.
\end{equation}
To prove this, consider the modes $m=l$ in the double Fourier representation. Then
\[
 k(-m\sqrt{1-t}+m\sqrt t)\longrightarrow\sqrt2ms,
 \qquad t=1/2+s/k,
\]
uniformly on compact $s$-intervals, so their amplitudes tend to $c_m^2e^{-m^2s^2}$.
For each fixed $m\ne l$, the damping factor tends to zero uniformly on the interval. Uniform absolute summability of the coefficients allows us
to truncate the modes and then control the tail. The difference between $M_{k,t}$ and
\[
 \sum_{m\ne0}c_m^2e^{-m^2s^2}
          e^{ikm(\sqrt t+\sqrt{1-t})z}
\]
tends to zero uniformly in $(s,z)$. The absolute differences between distinct retained frequencies tend to infinity uniformly for $|s|\le S_0$,
so the squared Gaussian $L^2$ norm tends uniformly to $\sum|c_m|^4e^{-2m^2s^2}$.
The linear output, normalization, and standardization errors remain $o(\epsilon_k^2)$ or
$o(\epsilon_k^4)$ on this compact interval. The same entropy Taylor formula proves~\eqref{eq:periodic-central-profile}.

In particular, for each fixed $s\ne0$ and sufficiently large $k$,
\[
 F_k(1/2+s/k)>F_k(1/2).
\]
For $u=\cos$, the profile is exactly $e^{-2s^2}/16$.
Thus the entropy variation near the midpoint has width of order $k^{-1}$ and depth of order $\epsilon_k^4$.
Although~\eqref{eq:periodic-central-profile} holds uniformly on compact $s$-intervals, its error is not controlled relative to $s^2$. It therefore does not determine the sign of the midpoint second derivative for finite $k$.

Equal weighting retains perturbations of the same frequency in the two inputs, whereas unequal weighting requires matching integer frequency ratios.
For instance, when $u=\cos$, the normalized midpoint entropy deficit tends to $1/16$, while the limit at every fixed interior nonmidpoint weight is zero.
The threshold for these convergences depends on the weight, and convergence is not uniform near the endpoints; the entropy deficit of the input density at an endpoint is typically of order $\epsilon_k^2$.

\section{Strictly positive midpoint curvature}\label{app:midpoint-curvature}
An entropy value at the midpoint below a nonmidpoint value does not by itself make the midpoint a local minimum; nor does the preceding scaling limit justify interchanging second derivatives.
We compute midpoint curvature directly and prove that, for every $q>2-\sqrt2$, a Hermite family requiring no moment correction satisfies $F''(1/2)>0$.
Thus, in this smaller noise range, equal weighting is indeed a strict local minimum.
The key is to estimate both the output perturbation and its second derivative with respect to the weight. This construction also matches more Gaussian moments as the degree increases.
Let $\varphi_s$ be the centered Gaussian density of variance $s$, let $\varphi=\varphi_1$, and use probabilists' Hermite polynomials $H_n$. Throughout, $n$ tends to infinity through even integers at least $4$.
\begin{proposition}\label{thm:hermite-noise-counterexamples}
Define
\[
 h_n(x)=\sqrt2 e^{-x^2/2}\frac{H_n(\sqrt2 x)}{\sqrt{n!}},\qquad
 \epsilon_n=e^{-\sqrt n},\qquad f_n(x)=\varphi(x)(1+\epsilon_n h_n(x)).
\]
For sufficiently large $n$, $f_n$ is a strictly positive, smooth, even probability density of variance one with $(\log f_n)''<-1/2$. All its moments through order $n-1$ agree exactly with those of the standard Gaussian, and, for each fixed $m\ge0$,
\[
 \|f_n/\varphi-1\|_{C_b^m}\longrightarrow0.
\]
Let $X_n\sim f_n$, let $G$ be an independent standard Gaussian, and set
\[
 Z_{q,n}=\sqrt q X_n+\sqrt{1-q}G,\qquad 0<q\le1.
\]
The variables $Z_{q,n}$ retain the stated symmetry, unit variance, moment matching, and uniform convergence of relative densities and their fixed-order derivatives; for sufficiently large $n$ their densities are also strongly log-concave. Let $Z_{q,n}'$ be an independent copy and write
\[
 F_{q,n}(t)=h\bigl(\sqrt{1-t}Z_{q,n}+\sqrt tZ_{q,n}'\bigr).
\]
For every fixed $q>2-\sqrt2$ and sufficiently large $n$,
\[
 F_{q,n}''(1/2)>0.
\]
Thus equal weighting is a strict local minimum. Moreover, for fixed $q_0>2-\sqrt2$, one common sufficiently large degree works for all $q\in[q_0,1]$.
\end{proposition}

\begin{proof}
\leavevmode
\par\medskip\noindent\textbf{Density and moment conditions}\par\smallskip
Lemma~\ref{lem:hermite-bounds} with $a=1$ gives
\[
 \|h_n\|_\infty\le B_n=\sqrt{2e}(n+1)^{1/4},\qquad
 \|h_n''\|_\infty\le(2n+2)B_{n+2},
\]
and $\epsilon_n\|h_n\|_{C_b^m}\to0$ for every fixed $m$.
The logarithmic differentiation estimate of Section~\ref{sec:construction-main} therefore shows that $f_n$ is eventually strictly positive and
\[
 (\log f_n)''\le-1+
 \frac{\epsilon_n(2n+2)B_{n+2}}{1-\epsilon_n B_n}<-1/2.
\]

We also have the exact identity
\[
 \varphi h_n=\varphi_{1/2}(x)\frac{H_n(\sqrt2 x)}{\sqrt{n!}}.
\]
Hermite orthogonality makes its integral against every polynomial of degree less than $n$ vanish. Normalization, variance one, and all the stated moment matches are therefore exact. Evenness follows from even $n$.

After Gaussian smoothing the density ratio is $1+\epsilon_n T_{\sqrt q}h_n$, where $T_\rho$ is the Gaussian Ornstein--Uhlenbeck operator. Its spatial derivatives satisfy
$\|(T_{\sqrt q}h_n)^{(m)}\|_\infty\le q^{m/2}\|h_n^{(m)}\|_\infty$, so the relative convergence of every fixed order and strong log-concavity are even uniform for $q\in[0,1]$. Adding an independent Gaussian also directly preserves the low-order moment matches.

\par\medskip\noindent\textbf{Exact decomposition of the weighted convolution}\par\smallskip
Write
\[
 R_{j,s}(x)=\frac{\varphi_s(x)}{\varphi(x)}\frac{H_j(x/\sqrt s)}{\sqrt{j!}},\quad
 d=\frac q2,\quad s(w)=1-dw,\quad s_M=1-\frac q2.
\]
The Fourier transform of $\varphi_{1/2}H_n(\sqrt2 x)$, or Gaussian derivative convolution identities, gives the exact output density
\[
 p_{q,n,t}=\varphi(1+\epsilon_n L_t+\epsilon_n^2M_t),
\]
where
\begin{align}
 \mathcal A_w&=\left(\frac{qw}{2s(w)}\right)^{n/2}R_{n,s(w)},\qquad
 L_t=\mathcal A_t+\mathcal A_{1-t},\label{eq:hermite-L}\\
 M_t&=\left(\frac{q\sqrt{t(1-t)}}{2s_M}\right)^n
 \sqrt{\binom{2n}{n}}\,R_{2n,s_M}.\label{eq:hermite-M}
\end{align}
All densities and the required parameter derivatives have Gaussian times polynomial envelopes, and the density ratios are uniformly bounded away from zero, justifying the entropy differentiations below.

Write $s_L=1-q/4$, $L=L_{1/2}$, $M=M_{1/2}$, and $\ddot L=\partial_t^2L_t|_{1/2}$. Exchange symmetry of the independent copies gives $\partial_t p|_{1/2}=0$, while
\[
 \partial_t^2M_t|_{1/2}=-4nM.
\]
Define $c_L=(q/(4s_L))^{n/2}$. Direct differentiation of~\eqref{eq:hermite-L} gives
\begin{align}
 L&=2c_LR_{n,s_L},\\
 \ddot L&=2c_L\left[
 n(n-2)R_{n,s_L}
 -\frac{nd}{s_L}\sqrt{(n+1)(n+2)}R_{n+2,s_L}\right.\nonumber\\
 &\hspace{30mm}\left.
 +\frac{d^2}{4s_L^2}\sqrt{(n+1)(n+2)(n+3)(n+4)}R_{n+4,s_L}\right].
 \label{eq:hermite-Lsecond}
\end{align}
Alternatively, the Fourier transform shows that the multiplier for the second derivative of a single term relative to the original term is
$n(n-2)+nd\xi^2+d^2\xi^4/4$; the Fourier signs of the consecutive even Hermite terms give the displayed expression.

\par\medskip\noindent\textbf{Exponential comparison of the linear and quadratic terms}\par\smallskip
All norms below are in $L^2(\varphi)$. Put
$a_j=\binom{2j}{j}/4^j$ and $b=s/(2-s)$. Mehler's identity and termwise integration give the exact positive-term formula
\begin{equation}\label{eq:hermite-norm-sum}
 \|R_{j,s}\|_2^2
 =\frac1{\sqrt{s(2-s)}}\sum_{i=0}^j a_i a_{j-i}b^i.
\end{equation}
Its generating function is $[s(2-s)(1-r)(1-br)]^{-1/2}$. For $0<s<1$, we have $0<b<1$, so the sum lies between $a_j$ and 1, using $\sum_i a_i a_{j-i}=1$ for the upper bound. Induction also gives $a_j\ge1/(2\sqrt j)$ for $j\ge1$.

For fixed $q>0$, or uniformly for $q\in[q_0,1]$ with $q_0>0$, these estimates give
\begin{align*}
 \|L\|_2&\le C\left(\frac q{4s_L}\right)^{n/2},\\
 \|\ddot L\|_2&\le Cn^2\left(\frac q{4s_L}\right)^{n/2},\\
 \|M\|_2&\ge c n^{-1/2}\left(\frac q{2s_M}\right)^n.
\end{align*}
The last bound uses both $\binom{2n}{n}\ge4^n/(2\sqrt n)$ and~\eqref{eq:hermite-norm-sum}. The exponential rate of the key ratio is therefore determined by
\[
 \mathcal R(q)=\frac{q s_L}{s_M^2}
 =\frac{q(1-q/4)}{(1-q/2)^2}
\]
and we have
\[
 \frac{\|L\|_2}{\epsilon_n\|M\|_2}
 +\frac{\|\ddot L\|_2}{4n\epsilon_n\|M\|_2}
 \le C n^{3/2}e^{\sqrt n}\mathcal R(q)^{-n/2}.
\]
Simplification gives
\[
 \mathcal R(q)>1\quad\Longleftrightarrow\quad q>2-\sqrt2.
\]
Hence the right-hand side tends to zero in this range.

\par\medskip\noindent\textbf{The logarithmic remainder and strictly positive curvature}\par\smallskip
Write $r=\epsilon_n L+\epsilon_n^2M$ and $\ddot r=\epsilon_n\ddot L-4n\epsilon_n^2M$. The conditional expectation representation with three independent Gaussians gives, for all weights,
\[
 \|p_{q,n,t}/\varphi-1\|_\infty\le
 \eta_n:=2\epsilon_nB_n+\epsilon_n^2B_n^2\longrightarrow0.
\]
The output variance is always one, so $\int\varphi\ddot r=\int x^2\varphi\ddot r=0$. Thus
\[
 F_{q,n}''(1/2)=-\int\varphi\ddot r\log(1+r).
\]
Set
\[
 A_n=\frac{\|L\|_2}{\epsilon_n\|M\|_2},\qquad
 B_n^*=\frac{\|\ddot L\|_2}{4n\epsilon_n\|M\|_2}.
\]
Using
$|\log(1+r)-r|\le\eta_n|r|/[2(1-\eta_n)]$ and Cauchy--Schwarz, we obtain
\begin{equation}\label{eq:hermite-curvature-certificate}
 \frac{F_{q,n}''(1/2)}{4n\epsilon_n^4\|M\|_2^2}
 \ge 1-A_n-B_n^*(1+A_n)
 -\frac{\eta_n}{2(1-\eta_n)}(1+B_n^*)(1+A_n).
\end{equation}
Indeed, expanding $-\int\varphi\ddot r\,r$ and applying Cauchy--Schwarz to each of the three cross terms gives
\[
 \left|\frac{F_{q,n}''(1/2)}{4n\epsilon_n^4\|M\|_2^2}-1\right|
 \le A_n+B_n^*(1+A_n)
 +\frac{\eta_n}{2(1-\eta_n)}(1+B_n^*)(1+A_n).
\]
For $q>2-\sqrt2$, the quantities $A_n,B_n^*,\eta_n$ tend to zero. The right-hand side thus tends to zero, and the normalized curvature tends to one. Consequently,
\[
 F_{q,n}''(1/2)\sim4n\epsilon_n^4\|M\|_2^2>0.
\]
All estimates are uniform for $q\in[q_0,1]$ with $q_0>2-\sqrt2$, giving a common degree threshold. Positive curvature and symmetry in the weight make equal weighting a strict local minimum.

\end{proof}

This construction also gives an entropy comparison at fixed weights: if $q>2-\sqrt2$, $t\ne1/2$, and
\[
 \max(t,1-t)<\frac q{q^2-2q+2},
\]
then $F_{q,n}(t)>F_{q,n}(1/2)$ for sufficiently large $n$.
Indeed, equations~\eqref{eq:hermite-L}--\eqref{eq:hermite-norm-sum} give
\[
 \|L_t\|_2=o(\epsilon_n\|M\|_2),\qquad
 \frac{\|M_t\|_2}{\|M\|_2}=[4t(1-t)]^{n/2}\longrightarrow0,
\]
so the relative entropy comparison of Section~\ref{sec:comparison} applies.

\section{An interpolation inequality for Gaussian smoothing}\label{sec:interpolation}\label{app:interpolation}
The symmetric proof requires the midpoint quadratic term to dominate the linear term, producing the restriction $4/7$. Is this restriction specific to the chosen Hermite sequence, or does it reflect a more general estimate?
We prove a Gaussian interpolation inequality for all perturbations in $L^2(dx)$ and determine its exact parameter range, including the endpoint.
At the parameters of the main proof, the inequality shows that for $0<q\le4/7$, if the amplitude times the Lebesgue $L^2$ norm of the perturbation tends to zero, then the midpoint quadratic term is negligible relative to the linear term.
This identifies the range of this dominance argument; it does not itself give sufficient conditions for entropy concavity.
Two integration measures occur here. We distinguish the Lebesgue norm $\|\cdot\|_{L^2(dx)}$ from the Gaussian norm $\|\cdot\|_{L^2(\varphi\,dx)}$.
\begin{theorem}\label{thm:sharp-gaussian-interpolation}
Let $0<r<1$ and $0<s\le1$. There is a finite constant $C_{r,s}$ such that every $g\in L^2(\mathbb R,dx)$ satisfies
\begin{equation}\label{eq:sharp-gaussian-interpolation}
 \|T_{\sqrt s}g\|_{L^2(\varphi\,dx)}^2
 \le C_{r,s}\|g\|_{L^2(dx)}\|T_{\sqrt r}g\|_{L^2(\varphi\,dx)},
\end{equation}
if and only if
\begin{equation}\label{eq:sharp-gaussian-interpolation-region}
 s\le\sqrt{2r(1+r)}-r.
\end{equation}
\end{theorem}
The range includes equality at the endpoint. The proof is given at the end of this appendix.
Taking $r=q/2$ and $s=q$, condition~\eqref{eq:sharp-gaussian-interpolation-region} is equivalent to
$3q/2\le\sqrt{q(1+q/2)}$, or $0<q\le4/7$.

For $g\in L^2(dx)$, denote the midpoint linear and quadratic conditional expectations by
\[
 L(g)=2T_{\sqrt{q/2}}g,\qquad
 M(g,g)(x)=\mathbb E[g(G_1)g(G_2)\mid
  \sqrt{q/2}(G_1+G_2)+\sqrt{1-q}G_3=x].
\]
Contraction of conditional expectation gives
$\|M(g,g)\|_{L^2(\varphi\,dx)}\le\|T_{\sqrt q}g\|_{L^2(\varphi\,dx)}^2$: smooth the two independent copies separately and then condition on their equally weighted sum.
We obtain the following corollary.
\begin{corollary}\label{cor:quadratic-obstruction}
For $0<q\le4/7$,
\[
 \|M(g,g)\|_{L^2(\varphi\,dx)}
 \le\tfrac12 C_{q/2,q}\|g\|_{L^2(dx)}\|L(g)\|_{L^2(\varphi\,dx)}.
\]
If $g_n\ne0$, $\epsilon_n>0$, and $\epsilon_n\|g_n\|_{L^2(dx)}\to0$, then
$\epsilon_n^2\|M(g_n,g_n)\|_{L^2(\varphi\,dx)}/
 (\epsilon_n\|L(g_n)\|_{L^2(\varphi\,dx)})\to0$.
\end{corollary}
\begin{proof}
Apply the theorem to the preceding contraction bound.
Since $\varphi\le(2\pi)^{-1/2}$, we have $L^2(dx)\subset L^2(\varphi\,dx)$.
In the complete orthonormal basis $H_j/\sqrt{j!}$ of $L^2(\varphi\,dx)$, $T_\rho$ multiplies the $j$th coefficient by $\rho^j$.
For $\rho=\sqrt{q/2}>0$, the equality $T_\rho g=0$ forces every coefficient to vanish, hence $g=0$. The denominator is therefore nonzero for nonzero $g_n$.
Completeness can also be obtained from uniqueness of the Gaussian-weighted Fourier transform: if $g$ is orthogonal to all polynomials, the entire function
$z\mapsto\int g(x)\varphi(x)e^{zx}\,dx$ has all Taylor coefficients zero and is identically zero. Taking purely imaginary $z$ gives $g\varphi=0$.
Entire analyticity and differentiation under the integral follow from Cauchy--Schwarz and uniform bounds on Gaussian exponential moments over compact sets.
\end{proof}
Our moment-corrected Hermite sequence meets this small-norm condition: the original Hermite functions have degree-independent Lebesgue norms, and the fixed-order correction tends to zero exponentially.
Thus, for $0<q\le4/7$, the midpoint quadratic term is negligible relative to the linear term for sequences satisfying this small-norm condition.
This compares only the sizes of the terms; it does not determine the curvature of entropy as a function of the weight. Concavity for symmetric inputs at $q\le4/7$ remains open.

The Gaussian smoothing norm corresponds to a positive integral operator on $L^2(dx)$. Sufficiency in Theorem~\ref{thm:sharp-gaussian-interpolation} follows from an operator-order comparison of Gaussian kernels; necessity is tested on high-degree Hermite eigenfunctions.
\begin{proof}
\leavevmode
\par\medskip\noindent\textbf{The Gaussian kernel of the smoothing quadratic form}\par\smallskip
Let $A_s$ be the positive operator on $L^2(dx)$ satisfying
$\langle g,A_sg\rangle_{dx}=\|T_{\sqrt s}g\|_{L^2(\varphi\,dx)}^2$.
For $s<1$, its kernel is
\[
 A_s(x,y)=c_s e^{-\alpha_s(x^2+y^2)+\beta_sxy},\qquad
 c_s=\frac1{2\pi\sqrt{1-s^2}},\quad
 \alpha_s=\frac1{2(1-s^2)},\quad\beta_s=\frac{s}{1-s^2}.
\]
Indeed, the kernel is the bivariate standard Gaussian density with correlation $s$.
For $s=1$, $A_1$ is multiplication by $\varphi$.

Fix $r$ and put
\[
 a=(1-r^2)^{-1/2},\qquad
 \lambda=\frac{r}{1+\sqrt{1-r^2}},\qquad v=\sqrt\lambda,
 \qquad \Lambda_r=\sqrt{\frac{a}{2\pi(1+a)}}.
\]
Let $h_{n,a}(x)=\sqrt2e^{-ax^2/2}H_n(\sqrt{2a}x)/\sqrt{n!}$.
The Hermite functions normalized in Lebesgue measure,
\[
 e_{n,a}=h_{n,a}/\sqrt{J_a},\qquad J_a=2\sqrt{\pi/a},
\]
form an orthonormal basis, and
\begin{equation}\label{eq:gaussian-operator-diagonal}
 A_r e_{n,a}=\Lambda_r\lambda^n e_{n,a}.
\end{equation}
Summing $\sum_{n\ge0}\lambda^n e_{n,a}(x)e_{n,a}(y)$ by Mehler's formula and multiplying by $\Lambda_r$ gives
\[
 \Lambda_r\sqrt{a/\pi}(1-\lambda^2)^{-1/2}
 \exp\left\{-\frac{a(1+\lambda^2)}{2(1-\lambda^2)}(x^2+y^2)
                 +\frac{2a\lambda}{1-\lambda^2}xy\right\}.
\]
Using $r=2\lambda/(1+\lambda^2)$ and $a=(1+\lambda^2)/(1-\lambda^2)$, the three kernel parameters equal $c_r,\alpha_r,\beta_r$, respectively. This proves~\eqref{eq:gaussian-operator-diagonal}.
For completeness, if $g\in L^2(dx)$ is orthogonal to all $e_{n,a}$, then
$\int g(x)e^{-ax^2/2}x^n\,dx=0$ for every $n$.
The function $z\mapsto\int g(x)e^{-ax^2/2}e^{zx}\,dx$ is entire, since Cauchy--Schwarz gives uniform integrability on compact sets.
All its derivatives at zero vanish, so it is identically zero. Taking purely imaginary $z$ and using uniqueness of the Fourier transform gives $g=0$.

Let $B=A_r^{1/2}$. Equation~\eqref{eq:gaussian-operator-diagonal} and Mehler's formula give
\[
 B(x,y)=c_Be^{-\alpha_B(x^2+y^2)+\beta_Bxy},
\]
\[
 c_B=\sqrt{\Lambda_r}\sqrt{\frac a\pi}\frac1{\sqrt{1-v^2}},\qquad
 \alpha_B=\frac{a(1+v^2)}{2(1-v^2)},\qquad
 \beta_B=\frac{2av}{1-v^2}.
\]

\par\medskip\noindent\textbf{An operator-order comparison of Gaussian kernels}\par\smallskip
We will apply the following kernel comparison to the square-root operator. For $c_j>0$ and $2\alpha_j>\beta_j>0$ ($j=A,B$), write
\[
 K_j(x,y)=c_j e^{-\alpha_j(x^2+y^2)+\beta_jxy}\quad(j=A,B).
\]
Suppose
\begin{equation}\label{eq:gaussian-kernel-condition}
 \beta_A+2|\alpha_A-\alpha_B|\le\beta_B.
\end{equation}
Then, in the positive operator order,
\begin{equation}\label{eq:gaussian-kernel-comparison}
 K_A\preceq\frac{c_A}{c_B}\frac{\beta_B}{\beta_A}K_B.
\end{equation}
The notation $\preceq$ denotes order of quadratic forms.

To prove this, consider the Hilbert space of entire functions on the complex plane, the Fock space, with norm
\[
 \|f\|_{\mathcal F}^2=\frac1\pi\int_{\mathbb C}|f(z)|^2e^{-|z|^2}\,dA(z).
\]
If $f(z)=\sum_{n\ge0}a_nz^n$, Parseval's identity on circles followed by radial integration gives
$\|f\|_{\mathcal F}^2=\sum_{n\ge0}|a_n|^2n!$; monotone convergence justifies interchanging the integrals of nonnegative terms.
Thus $z^n/\sqrt{n!}$ is an orthonormal basis. Summing gives the reproducing kernel
$\sum_{n\ge0}z^n\overline w^{\,n}/n!=e^{z\overline w}$, and Cauchy--Schwarz gives boundedness of point evaluation.
The weighted composition operator used here falls under Le's boundedness criterion~\cite[Theorem 2.2]{Le} and the critical quadratic-weight case of Carroll--Gilmore~\cite{CG}. For the real quadratic weights needed here, we give a direct norm estimate including the boundary.
Put $\delta=\alpha_A-\alpha_B$ and $\rho^2=\beta_A/\beta_B$, and consider
\[
 (Cf)(z)=e^{-\delta z^2/\beta_B}f(\rho z).
\]
After the substitution $w=\rho z$, condition~\eqref{eq:gaussian-kernel-condition} gives
\begin{align*}
 \|Cf\|_{\mathcal F}^2
 &=\frac1{\pi\rho^2}\int_{\mathbb C}|f(w)|^2
 \exp\!\left[-\frac{|w|^2}{\rho^2}
       -\frac{2\delta\operatorname{Re}(w^2)}{\beta_B\rho^2}\right]dA(w)\\
 &\le\rho^{-2}\|f\|_{\mathcal F}^2.
\end{align*}
The estimate remains valid at equality, since the quadratic decay coefficients in both real directions are at least one.
Write $k_w(z)=e^{z\overline w}$. The reproducing property gives
$C^*k_w=e^{-\delta\overline w^2/\beta_B}k_{\rho w}$.
For real $x$, take $\xi_x=e^{-\alpha_Bx^2}k_{\sqrt{\beta_B}x}$. Then
\[
 \langle\xi_y,\xi_x\rangle=K_B(x,y)/c_B,\qquad
 \langle C^*\xi_y,C^*\xi_x\rangle=K_A(x,y)/c_A.
\]
Since $\|C^*\|^2\le\rho^{-2}$, the Gram matrices of finite linear combinations satisfy the claimed comparison.
Approximate the quadratic forms of compactly supported continuous functions by Riemann sums, then use density in $L^2$ to obtain~\eqref{eq:gaussian-kernel-comparison}.
The kernels here satisfy $2\alpha_j>\beta_j>0$ and are therefore Hilbert--Schmidt, so their quadratic forms are continuous on $L^2$.

\par\medskip\noindent\textbf{Sufficiency}\par\smallskip
In the asserted parameter range, the smoothing quadratic form is controlled by the square-root quadratic form. A single Cauchy--Schwarz inequality then yields the interpolation bound.
Substitute $K_A=A_s$ and $K_B=B$. Condition~\eqref{eq:gaussian-kernel-condition} is equivalent to
\[
 \frac1{1-s}\le a\frac{1+v}{1-v},\qquad
 a\frac{1-v}{1+v}\le\frac1{1+s}.
\]
Write
\[
 r=\frac{2v^2}{1+v^4},\qquad a=\frac{1+v^4}{1-v^4}.
\]
The first inequality is equivalent to
\[
 s\le1-\frac{1-v}{a(1+v)}
   =\frac{2v(1-v+v^2)}{1+v^4}
   =\sqrt{2r(1+r)}-r.
\]
The second gives the upper bound
$2v(1+v+v^2)/(1+v^4)$, strictly larger than the first. Thus the required range is precisely
\eqref{eq:sharp-gaussian-interpolation-region}.
In this range,~\eqref{eq:gaussian-kernel-comparison} and Cauchy--Schwarz give
\begin{align*}
 \|T_{\sqrt s}g\|_{L^2(\varphi\,dx)}^2
 &\le C_{r,s}\langle g,A_r^{1/2}g\rangle_{dx}\\
 &\le C_{r,s}\|g\|_{L^2(dx)}\|A_r^{1/2}g\|_{L^2(dx)}
 =C_{r,s}\|g\|_{L^2(dx)}\|T_{\sqrt r}g\|_{L^2(\varphi\,dx)},
\end{align*}
with the explicit choice
$C_{r,s}=(c_s/c_B)(\beta_B/\beta_s)$. This constant suffices to prove sufficiency of the parameter range.

\par\medskip\noindent\textbf{Necessity}\par\smallskip
Outside this range, high-degree Hermite functions make the ratio of the two sides tend to infinity, so no perturbation-independent constant exists.
Again take $a=(1-r^2)^{-1/2}$ and use $g=h_{n,a}$, with $n$ tending to infinity through even integers.
The squared Lebesgue norm is always $J_a$, and
\[
 \|T_{\sqrt r}h_{n,a}\|_{L^2(\varphi\,dx)}^2
 =J_a\Lambda_r\lambda^n.
\]
If $s>\sqrt{2r(1+r)}-r$, then necessarily $s>r$. In the linear generating function above,
\[
 \tau_a(s)=\frac{1-a+as}{1+a-as}
 >\nu_a(s)=\frac{a-1+as}{1+a+as}>0,
 \qquad \tau_a(s)>v.
\]
The last strict inequality is equivalent to
$s>1-(1-v)/(a(1+v))$. Let $A_n=4^{-n}\binom{2n}{n}$. By nonnegativity of the generating coefficients, retaining
the $n$th term of $(1-\tau_a(s)z)^{-1/2}$ and the constant term of the other factor gives
\[
 \|T_{\sqrt s}h_{n,a}\|_{L^2(\varphi\,dx)}^2
 \ge\frac{2A_n\tau_a(s)^n}{\sqrt{(1+a)^2-a^2s^2}}
 \ge c_{r,s}n^{-1/2}\tau_a(s)^n,
\]
The last step uses the Stirling lower bound for the central binomial coefficient.
Thus the ratio of the left-hand side of~\eqref{eq:sharp-gaussian-interpolation} to the right-hand side without its constant
grows at least as $n^{-1/2}(\tau_a(s)/v)^n$ and tends to infinity. No uniform constant exists.
For $s=1$, the generating function remains valid and $\tau_a(1)=1>v$, giving the same conclusion.

\end{proof}

\section{An explicit example with noise variance fraction 17/40}\label{sec:finite-certificate}
The preceding asymptotic construction produces counterexamples at sufficiently high degree. We finish with fully specified parameters that turn this conclusion into a strict finite inequality.
We construct a strongly log-concave counterexample with noise variance fraction $17/40$ and give an explicit positive lower bound for the entropy gap.
The example demonstrates a finite realization; the general existence result is proved by the analytic estimates in the main text.
Take
\[
 q=23/40,\quad t=503/1000,\quad n=2^{24},\quad
 a=7/(3\sqrt5),\quad \epsilon=e^{-4096}.
\]
Correct only mass and variance, so $K=1$. The original zeroth and second moments, denoted by $m_0,m_2$, are exactly
\[
 m_0=\sqrt{\frac2{1+a}}\frac{\sqrt{n!}}{2^{n/2}(n/2)!}b^{n/2},
 \qquad
 m_2=m_0\left(\frac1{1+a}+\frac{2an}{(1+a)^2b}\right).
\]
Set $c_{n,0}=m_0$ and $c_{n,1}=\sqrt2(m_2-m_0/2)$, and put
\[
 f=\varphi\{1+\epsilon(h_n-c_{n,0}\psi_0-c_{n,1}\psi_1)\}.
\]
\begin{proposition}\label{prop:finite-noise}
The density $f$ above is strictly positive, even, smooth, and of variance one. Both $f$ and $\mathcal S_qf$ have second logarithmic derivative less than $-1/2$.
Moreover, $\|f/\varphi-1\|_\infty<2.013\cdot10^{-1777}$, and
\[
 F_{\mathcal S_qf}(503/1000)=F_{\mathcal S_qf}(497/1000)
 >F_{\mathcal S_qf}(1/2)+10^{-13896577}.
\]
\end{proposition}
\begin{proof}

Normalization and the moment conditions follow from the correction formulas in Section~\ref{sec:construction-main}.
Use $d,e,c,\kappa,\Delta$ as in Section~\ref{sec:quadratic}, and choose the integer counts
\[
 k_d=541411,\qquad k_e=2009132,\qquad k_c=6842631.
\]
These satisfy $2k_d+k_e+2k_c=n$. Let
\[
 A^2=\frac{4(n!)^2}{\Delta}
       \frac{(d/2)^{4k_d}e^{2k_e}c^{4k_c}}{(k_d!)^4(k_e!)^2(k_c!)^4},
 \qquad B^2=\frac4\Delta\frac{(n!)^2}{(2n)!}(2\kappa_t)^{2n},
\]
Then $\|M_{1/2}\|_2\ge A>0$ and $\|M_t\|_2\le B$.
In this appendix, $A,B$ denote only these norm bounds.
Define the linear upper bounds
\[
 U_s^2=\frac8{\sqrt{(1+a)^2-a^2q^2w_s^2}}\tau(qw_s)^n,
 \qquad w_s=\max(s,1-s),\quad s\in\{1/2,t\}.
\]
Also write
\begin{gather*}
 R=|c_{n,0}|B_0+|c_{n,1}|B_2,\qquad H=B_n+R,\\
 J=a(2n+2)B_{n+2}+2|c_{n,0}|B_2+6|c_{n,1}|B_4,\\
 E=2B_nR+R^2,\qquad \eta=2\epsilon H+\epsilon^2H^2,\\
 \ell=\frac EA+\frac{2R}{\epsilon A}+\frac{U_{1/2}}{\epsilon A},\qquad
 v=\frac EA+\frac{2R}{\epsilon A}+\frac{U_t}{\epsilon A}+\frac BA.
\end{gather*}
Lemma~\ref{lem:hermite-bounds} gives $\|\widetilde h_n\|_\infty\le H$ and $\|\widetilde h_n''\|_\infty\le J$.
If $\epsilon H<1$ and $-1+\epsilon J/(1-\epsilon H)<-1/2$, then $f$ and its Gaussian smoothing are strongly log-concave.
If also $\ell<1$ and $\eta<1$, the triangle inequality and Lemma~\ref{lem:entropy-quadratic} give
\begin{equation}\label{eq:finite-gap}
 F(t)-F(1/2)\ge\epsilon^4A^2\,
 \frac12\left\{\frac{(1-\ell)^2}{1+\eta}-\frac{v^2}{1-\eta}\right\}.
\end{equation}

Evaluating these expressions with 512-bit Arb interval arithmetic~\cite{Arb} gives the following strict upper bounds:
\[
\begin{array}{c|c}
\text{Expression}&\text{Strict upper bound}\\\hline
 \epsilon H&2.013\cdot10^{-1777}\\
 U_{1/2}/(\epsilon A)&3\cdot10^{-35499}\\
 U_t/(\epsilon A)&9\cdot10^{-6260}\\
 B/A&3\cdot10^{-54}\\
 E/A&10^{-7080183}\\
 2R/(\epsilon A)&5\cdot10^{-7078407}
\end{array}
\]
In addition, $-1+\epsilon J/(1-\epsilon H)<-1/2$, and
\[
 \frac12\left\{\frac{(1-\ell)^2}{1+\eta}-\frac{v^2}{1-\eta}\right\}>
 \frac{4999}{10000},\qquad
 \log_{10}\!\left(\epsilon^4A^2\frac{4999}{10000}\right)>-13896577.
\]
Factorials are evaluated through $\log\Gamma(k+1)$, and the inequalities follow from the corresponding directed interval endpoints.
The expressions and interval endpoints are supplied with the paper.
Substitution into~\eqref{eq:finite-gap} yields
\[
 F(503/1000)=F(497/1000)>F(1/2)+10^{-13896577}.
\]
\end{proof}

\end{document}